\documentclass[lettersize,journal]{IEEEtran}
\usepackage{algorithmic}
\usepackage{algorithm}
\usepackage{array}
\usepackage{textcomp}
\usepackage{stfloats}
\usepackage{url}
\usepackage{verbatim}
\usepackage{graphicx}
\usepackage{cite}
\usepackage{amssymb}
\usepackage{ifthen}
\usepackage[utf8]{inputenc}
\usepackage{latexsym, dsfont}
\usepackage{amsmath,amsthm}
\usepackage[inline]{enumitem}
\usepackage{comment}
\usepackage{siunitx}
\usepackage{graphicx}
\usepackage{caption}
\usepackage{hyperref}   
\usepackage[capitalize]{cleveref}
\usepackage{todonotes}
\usepackage[caption=false,font=footnotesize]{subfig}
\usepackage{mathtools}
\usepackage{tikz}
\newboolean{PGFPlots}
\setboolean{PGFPlots}{true}

\newtheorem{thm}{Theorem}

\newtheorem{corol}{Corollary}

\usepackage{xcolor}
\usepackage{balance}
\usetikzlibrary{arrows,shapes,decorations,decorations.pathreplacing,decorations.pathmorphing,intersections,calc,fit}

\ifthenelse{\boolean{PGFPlots}}{
\usepackage{pgfplots}
\usepgfplotslibrary{units,groupplots}
\usetikzlibrary{external}
\tikzexternalize
\tikzset{external/only named=true}
}{}

\newcommand{\dupD}{\delta^\uparrow}
\newcommand{\ddoD}{\delta^\downarrow}

\newcommand{\dsu}{\delta_S^\uparrow}
\newcommand{\dsd}{\delta_S^\downarrow}

\newcommand{\figPath}[1]{figures/#1}

\newcommand{\NOR}{\texttt{NOR}}
\newcommand{\NAND}{\texttt{NAND}}

\newcommand{\cg}{\texttt{C}}
\newcommand{\AOIgate}{\texttt{AOI}}

\newcommand{\dmin}{\delta_{\mathrm{min}}}

\newcommand{\vth}{V_{th}}

\newcommand{\vdd}{V_{DD}}
\newcommand{\gnd}{\textit{GND}}
\newcommand{\vout}{V_{out}}

\newcommand{\dd}{\mathrm{d}}

\newcommand{\invt}{\texttt{InvTool}}

\newcommand{\nmos}{\texttt{nMOS}}
\newcommand{\pmos}{\texttt{pMOS}}

\newcommand{\ohm}{(OHM)}

\newcommand{\on}{\mbox{\emph{on}}}
\newcommand{\off}{\mbox{\emph{off}}}

\newboolean{conference}
\setboolean{conference}{true}  

\begin{document}

\title{A Hybrid Delay Model for Interconnected Multi-Input Gates\thanks{
The research work of Arman Ferdowsi was funded by the Austrian Science Fund (FWF) project DMAC [10.55776/P32431]. Matthias Függer's research was supported by the ANR DREAMY (ANR-21-CE48-0003). Josef Salzmann was funded by the the Austrian Federal Ministry for Digital and Economic Affairs, the National Foundation for Research, Technology and Development and the Christian Doppler Research Association.}}

%
%
%

\author{
  \IEEEauthorblockN{
    Arman Ferdowsi\IEEEauthorrefmark{1}, 
    Matthias Függer\IEEEauthorrefmark{3},
        Josef Salzmann\IEEEauthorrefmark{2}, 
        Ulrich Schmid\IEEEauthorrefmark{1}
  }

  \IEEEauthorblockA{\IEEEauthorrefmark{1}TU Wien, Embedded Computing Systems Group, Vienna, Austria \\ 
    \{aferdowsi, s\}@ecs.tuwien.ac.at}
    
      \IEEEauthorblockA{\IEEEauthorrefmark{3}CNRS, LMF, ENS Paris-Saclay, Université Paris-Saclay, Paris, France \\ 
    mfuegger@lmf.cnrs.fr} 
    
  \IEEEauthorblockA{\IEEEauthorrefmark{2}TU Wien, CD-Laboratory for Single Defect Spectroscopy at the Institute for Microelectronics, Vienna, Austria \\ 
    salzmann@iue.tuwien.ac.at}

}



\maketitle

\begin{abstract}

Dynamic digital timing analysis is a less accurate but fast alternative to highly accurate but slow analog simulations of digital circuits. It relies on gate delay models, which allow the determination of input-to-output delays of a gate on a per-transition basis. Accurate delay models not only consider the effect of preceding output transitions here but also delay variations induced by multi-input switching (MIS) effects in the case of multi-input gates.
Starting out from a first-order hybrid delay model for CMOS 
two-input \NOR\ gates, we develop a hybrid delay model for Muller \cg\ gates 
and show how to augment these models and their analytic delay formulas by a first-order 
interconnect. Moreover, we conduct a systematic evaluation of the resulting modeling 
accuracy: Using SPICE simulations, we quantify the MIS effects on the gate delays under various wire lengths, load capacitances, and input strengths for two different CMOS technologies, comparing these results to the predictions of appropriately parameterized versions of our new gate delay models. Overall, our experimental results reveal that they capture all MIS effects with a surprisingly good accuracy despite being first-order only.
\end{abstract}


\newcommand{\NEW}[1]{{\color{blue}#1}}

\section{Introduction}
\label{sec:intro}

\emph{Digital} timing analysis techniques are indispensable in modern circuit design, 
as they enable the validation of large designs: 
Thanks to the elaborate \emph{static timing analysis} (STA) techniques available
for digital timing analysis nowadays, which employ elaborate
models like CCSM~\cite{Syn:CCSM} and 
ECSM~\cite{Cad:ECSM} that facilitate an accurate corner-case analysis of the gate delays along the critical paths of a digital circuit,
worst-case as well as best-case delays can be determined
accurately and very fast.

The major shortcoming of classic STA is its inability to explicitly take into account
PVT variations and dynamic effects like slew-rate variations, crosstalk, and \emph{multi-input switching} (MIS) 
effects \cite{CS96:DAC,CGB01:DAC}. More precisely, corner-case analysis techniques can incorporate 
such effects only by
considering the overall worst-case resp.\ best-case scenario, which usually generates (way) too
conservative delay estimates. Whereas approaches like propagation windows \cite{SRC15:TODAES} have been
proposed as a means for reducing the resulting excessive margins, they cannot
be eliminated completely.
\emph{Statistical} static timing analysis (SSTA) techniques 
\cite{BCSS08:TCAD,FP09:VLSI} have hence been developed to mitigate this problem.
SSTA is based on some statistical models of the variabilities and computes statistical delay estimates based on those;
existing approaches \cite{ADB04:DAC,SZ05:ICCAD,YLW05,FTO08:DAC,TZBM11,SRPW20:DAC} differ 
in the particular statistical models used. 
Obviously, this works very well for effects like process variations,
where an accurate statistical model can be provided, but not for effects
that depend on the actual trace of the signals generated by a circuit.

Indeed, neither classic STA nor SSTA can take into account the actual
trace history for a given signal transition. The delay of a gate
stimulated by this transition may depend strongly on this history, however: Even
for a single-input-single-output gate like an inverter, an input
transition causing the output to change shortly after the previous
output change will have a shorter delay than an output change that happens
only after the previous one has saturated (this variability has been called
\emph{drafting effect} in \cite{CharlieEffect}). For multi-input gates,
gate delays may also vary when transitions at \emph{different} inputs
happen in close proximity (such MIS effects are also known as
\emph{Charlie effects}, named after Charles Molnar, who identified
their causes in the 70th of the last century).

Consequently, in the case of a timing violation reported by STA
in some critical part of a circuit, which may in fact be a false negative
originating in excessive margins, verifying the correct operation requires a detailed analysis
of the actual signal traces generated by the circuit.
This need can be nicely exemplified by means of the token-passing ring described
and analyzed 
by Winstanley et al. in \cite{CharlieEffect}, for example: This asynchronous
circuit implements a ring oscillator made up of stages consisting of a
2-input Muller \cg\ gate, the inputs of which are connected to the previous resp.\ next
stage. The authors uncovered that the ring exhibits two 
modes of operation, namely, burst behavior versus evenly spaced output transitions, 
which can even alternate over time.
The actual mode depends on some subtle interplay between the drafting
effect and the Charlie effect of the \cg\ gates in the circuit. Hence,
in order to analyze the behavior of the ring, the timing relation of
\emph{individual} transitions need to be traced throughout the
circuit.

Facilitating this is the purpose of \emph{dynamic timing analysis} techniques,
which are supported by virtually all modern circuit design tools nowadays.
Analog simulations, e.g., using SPICE \cite{NP73:spice}, are the golden
standard here. Unfortunately, however, analog simulation times are prohibitively expensive even for moderately large circuits and short signal traces, as the
dimension of the system of differential equations that need to be solved
numerically increases with the number of transistors. By constrast,
\emph{digital} dynamic timing analysis techniques rest on \emph{delay models} that provide gate delay
estimations on a per-transition basis. Suitable gate delay models allow fast correctness validation and accurate performance and power estimation \cite{najm1994survey} of a large circuit even at early stages of the development.

Compared to (S)STA, the-state-of-the-art in \emph{accurate} digital dynamic timing
  analysis is much less developed, however. Industrial tools like ModelSim
  offer only simple
pure and inertial delay models \cite{Ung71} here: CCSM and ECSM models
are typically used here for computing (constant) gate delays, which are
subsequently employed 
for \emph{every} transition occurring in the timing simulation. Since this approach
inherently exhibits no (pure delays) resp.\ almost no (inertial delays)
history dependency, it can neither model the drafting or MIS effects.

The present paper\footnote{A preliminary version of this paper has appeared
at DSD'23 (where it was nominated for the best paper award). This conference version was
based on the results of \cite{ferdowsi2023accurate}, however, which are piecewise approximations
of the gate delay. The interconnect extension described in the present paper is based
on the explicit delay formulas developed in \cite{ferdowsi2024faithful}. An extended version of the present paper can be found in \cite{ferdowsi2024hybrid}.} 
contributes to advancing the state-of-the-art in digital
dynamic timing analysis by providing fast delay models that
accurately capture all
MIS effects. More specifically, we augment the thresholded hybrid model 
proposed for ``naked''
2-input \NOR\ gates in \cite{ferdowsi2024faithful} to also
incorporate the interconnect to the successor gates, and develop an analogous
model for Muller \cg\ gates as well. The result of our efforts are simple and
efficiently parametrizable analytic delay formulas, which allow to compute
the delay $\delta(\Delta)$ for a given transition with \emph{input separation
time} $\Delta$ (which is the time difference to the closest transition of
the other input) by means of a single function evaluation.

In order to evaluate how accurately our models capture MIS effects for
gates with different interconnects, driving strengths, output loads, and
implementation technologies, we parametrize our model to match 3 specific
delay values ($\delta(0)$, $\delta(\infty)$ and $\delta(-\infty)$) of 
the given gate in its specific environment, and then compare
the model's delay predictions $\delta(\Delta)$ to SPICE simulation data, for
arbitrary values of $\Delta$. Overall, our experimental results reveal
that our model captures all MIS effects in any setting with a surprisingly 
good accuracy. This is quite surprising, given that they are based on 
first-order thresholded hybrid models (see \cref{sec:DDTA}) only, which
in turn are instrumental for developing analytic delay formulas that facilitate
fast digital timing analysis. 
We need to stress already here, however, that our models
are neither suited nor intended for being used in the context of STA frameworks; indeed,
due to their simplicity, our models are not competitive in terms of accuracy to the 
models surveyed in \cref{sec:MIS}.

\medskip

\noindent
\emph{Detailed contributions:}

\begin{enumerate}
\item[(1)] We augment the model from \cite{ferdowsi2024faithful} by an RC-type
interconnect and determine analytic expressions
for the trajectories and the resulting delays. The choice for a simple RC-type
model, as opposed to higher-order models, is motivated by the goal of an analytically solvable and simple-to-compute first-order model.

\item[(2)] We extend the parametrization procedure from \cite{ferdowsi2024faithful}
  to also determine the additional interconnect-related model parameters. Our extension
hence preserves the important advantage of analytic delay formulas, which is a strikingly
simple gate characterization procedure: For characterizing a 2-input \NOR\ gate, we only
need to plug-in three characteristic MIS delay values ($\delta(0)$, $\delta(\infty)$ and 
$\delta(-\infty)$ of the to-be-characterized gate into some parametrization functions, which
compute all the model parameters needed for matching those delays. Note that we even managed to incorporate
and additional interconnect pure delay $\dmin$ into our model.
  
\item[(3)] Based on our findings for the \NOR\ gate in (1) and (2), we develop a thresholded 
hybrid model for the Muller \cg\ gate and its parametrization procedure. However, due to space limitations, we had to defer the detailed exposition to a an extended version \cite{ferdowsi2024hybrid} of the present paper.

\item[(4)] We conduct a series of simulations to determine the accuracy of our augmented models
  for interconnected \NOR\ gates. In our evaluation,
we consider two different CMOS technologies (\SI{15}{\nm} and
\SI{65}{\nm}), and vary input driving strength (i.e., slew rate), wire length, 
load capacitance, wire resistance, and
wire capacitance. Using SPICE simulations, we determine the actual
gate delays for different values of $\Delta$, and compare those
to the predictions $\delta(\Delta)$ of appropriately parametrized versions
of our gate delay formulas. A preliminary implementation of our 
models in the discrete event simulation-based Involution Tool \cite{OMFS20:INTEGRATION} 
is used to demonstrate a speedup of several orders of magnitude compared to SPICE.
\end{enumerate}

\textbf{Paper organization:}
In \cref{sec:relwork}, we provide an account of related work.
In \cref{sec:background}, we summarize the cornerstones
of the thresholded hybrid model for \NOR\ gates presented in \cite{ferdowsi2024faithful}. In \cref{Sec:IntModel} resp.\ \cref{Sec:IntModelC}, we present our interconnect-augmented model, its parametrization, and its experimental accuracy evaluation for \NOR\ resp.\ \cg\ gates. Some concluding remarks and directions of future work in \cref{Sec:con} round-off our paper.

\section{Related Work}
\label{sec:relwork}

In this section, we briefly report on related work on delay models 
that capture MIS effects (\cref{sec:MIS}) and on accurate digital 
dynamic timing analysis (\cref{sec:DDTA}).

\subsection{Multi-input switching (MIS) effects}
\label{sec:MIS}

Consider the CMOS implementation of a \NOR\ gate shown in \cref{fig:nor_CMOS}, which consists of two serial \pmos\ ($T_1$ and $T_2$) for charging the load capacitance $C$ (producing a rising output transition), and two parallel \nmos\ transistors ($T_3$ and $T_4$) for discharging it (producing a falling one). 
When an input experiences a rising transition, the corresponding \nmos\ transistor closes while the corresponding \pmos\ transistor opens, so $C$ will be discharged.
If both inputs $A$ and $B$ experience a rising transition at the same time, 
$C$ is discharged twice as fast. Since the gate delay depends on the 
discharging speed, it follows that the falling output delay 
$\dsd(\Delta)$ increases (by almost 30\% in the example shown in \cref{corFig3}) when the \emph{input separation time} 
$\Delta=t_B-t_A$ increases from 0 to $\infty$ or decreases from 0 to $-\infty$.
Obviously, this MIS effect gets more pronounced for \NOR\ gates with a
  larger fan-in, as the number of parallel $\nmos$ transistors increases.

For falling input transitions, the behavior of the \NOR\ gate is quite different: \cref{corFig5} shows that the MIS effects lead to a moderate decrease of the rising output delay $\dsu(\Delta)$ when $|\Delta|$ goes from 0 to $\infty$, which is primarily caused by capacitive coupling and the parasitic capacitances at the nodes between the \pmos\ transistors. Note that the resulting MIS effect is much
less pronounced than the one for the falling output delay.

\begin{figure}[t!]
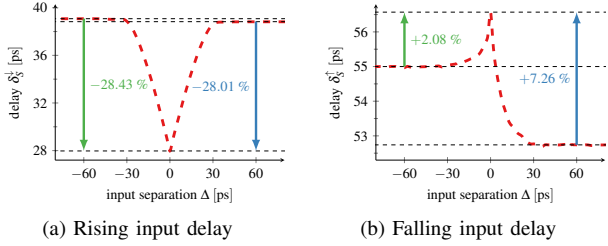

  \centering
  \subfloat[Rising input delay]{
    \includegraphics[width=0.43\linewidth]{\figPath{nor2_out_down_charlie_15nm_colored.pdf}}%
    \label{corFig3}}
  \hfil
  \subfloat[Falling input delay]{
    \includegraphics[width=0.43\linewidth]{\figPath{nor2_out_up_charlie_15nm_colored.pdf}}%
    \label{corFig5}}
  \caption{MIS effects in the measured delay of a $15$nm technology
    CMOS \NOR\ gate. $\Delta=t_B-t_A$ is the input separation time between
  effective signal transitions at the inputs $A$ and $B$.}\label{fig:Charlie15nmSim}
\end{figure}

\medskip

Quite a number of different models capable of capturing MIS effects have
been proposed in the literature, most of them in the context of STA
approaches. In~\cite{CS96:DAC}, Chandramouli and Sakallah resorted to
macromodels, i.e., blackbox functions involving delay-relevant input
parameters like load capacitance as well as the input transition time(s),
which compute the gate delay. They also show how to compose 2-input macromodels
to get $n$-input macromodels. Their approach achieves a delay and slew rate
prediction accuracy in the $1 \dots 10$\% range.

In \cite{CGB01:DAC}, the authors develop empirical delay formulas that
cover MIS effects for 2-input gates, and use curve fitting (based on detailed
simulation data, for every gate) for determining the appropriate parameters. They also show how to incorporate their model into STA, resorting to incremental timing refinement (ITR) to reduce the margins. 
A similar approach has been advocated in~\cite{SKJPC09:ISOCC}, where
a quadratical polynomial is used for the delay formula. In \cite{SRC15:TODAES},
Subramaniam, Roveda, and Cao used a piecewise linear delay function for
this purpose.
For model parametrization of a 2-input gate, only three representative
delay values (minimum, maximum and some intermediate value) are needed,
which makes gate characterization orders of magnitude faster than in
other approaches.
The authors use propagation windows to reduce the margins resulting
from using their modeling approach in STA.

A machine-learning-based MIS modeling approach has been proposed
in~\cite{RS21:TCAD}. Whereas it is distinguished by its very
good accuracy, which is in the few \%-range, its downside is the
inevitable per-gate training requirement.

\medskip

All the MIS models surveyed above target the delay functions 
themselves. A very different type of delay model is obtained
by developing a simplified model of a gate and determining the 
resulting delay function. In the approach proposed in \cite{AKMK06:DAC},
Amin et~al.\ used an analog model that consists of suitably
defined non-linear resistors and capacitors at each pin of 
a gate. Circuits are composed by composing the appropriate
models of the gates along the circuit's paths. Accuracy validation 
experiments revealed an accuracy in the $1 \dots 10$\% range.

A similar type of MIS delay models has been developed in the context
of digital dynamic timing analysis, which will be surveyed in the
following subsection. In sharp contrast to the model of 
\cite{AKMK06:DAC}, which deals with analog signals, the 
models described there belong to the class of thresholded 
hybrid models and hence process and generate digital signals.

\subsection{Digital dynamic timing analysis}
\label{sec:DDTA}

Digital delay models suitable for \emph{accurate} dynamic timing analysis must at least be single-history \cite{FNS16:ToC}: In order to model the drafting effect mentioned in \cref{sec:intro}, a gate's input-to-output delay $\delta(T)$ must be dependent on the previous-output-to-input delay $T$. Non-trivial examples of single-history models are the \emph{delay degradation model} (DDM) \cite{BJV06}, which uses an exponential function for $\delta(T)$, and the \emph{involution delay model} (IDM) \cite{FNNS19:TCAD}, the delay function of which is a negative involution $-\delta(-\delta(T))=T$.

\begin{figure}[t]
  \centerline{
    \includegraphics[width=0.75\linewidth]{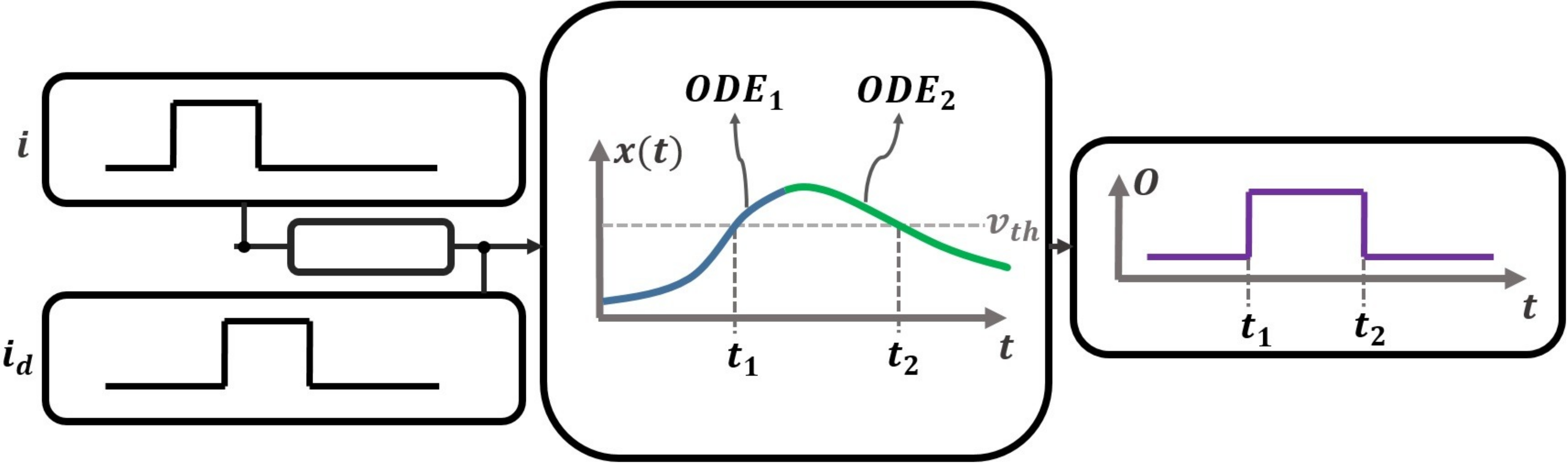}
  }
  \caption{Illustration of the thresholded hybrid system of an IDM channel, with a single input $i$ and output $o$.
  It comprises an (optional) pure delay shifter, producing $i_d$, and two ODEs governing some state signal $x(t)$
  that is digitized by a threshold voltage comparator to produce $o$. The active ODE is selected by the current 
  state of $i_d$, with mode switches that guarantee continuity of $x(t)$.}
  \label{fig:switched}
\end{figure}

Among all currently known delay models, the IDM is the only one that faithfully models glitch propagation in the canonical short-pulse filtration problem \cite{FNS16:ToC}. Another particularly compelling feature of the IDM is its
simplicity, which originates in the fact that it can be viewed as a simple 2-state thresholded first-order hybrid model \cite{FFNS23:HSCC}. As illustrated in \cref{fig:switched}, in such a system, the state of the digital input is used to select one among two \emph{ordinary differential equations} (ODEs) that govern some internal analog signal, the 
digitized version of which constitutes the digital output. This simplicity allows one to derive explicit analytic formulas for the IDM channel delays $\delta(T)$, which are instrumental for very fast digital timing simulation.

The IDM also comes with a publicly available discrete-event simulation framework,
the \emph{Involution Tool} \cite{OMFS20:INTEGRATION}, which 
also allows the evaluation of the accuracy of IDM delay predictions against SPICE-generated data and other delay 
models. Both measurements \cite{NFNS15:GLSVLSI} and simulations \cite{OMFS20:INTEGRATION} using 
the Involution Tool showed a very good accuracy for inverter chains and clock trees.
Since the simulation environment of the Involution Tool does not use
numerical integration like SPICE, but rather just invokes
the computation of the delay (a single function evaluation) once
per transition, it is orders of magnitude faster than analog
simulations. For example, in the case of the clock tree, a 
speedup of a factor of 250 has been obtained relative to SPICE
in terms of simulation times in \cite{OMFS20:INTEGRATION}.

For circuits also involving multi-input gates, however, the delay prediction
accuracy of the IDM degrades considerably \cite{OMFS20:INTEGRATION}. This
is not surprising since the single-input single-output IDM delay channels
are obviously incapable of varying the gate delay based on the input separation
time $\Delta$ between signal transitions at \emph{different} inputs.
Gate delay models that explicitly cover MIS effects 
are inevitable for mitigating this problem.

In \cite{FMOS22:DATE, ferdowsi2023accurate, ferdowsi2024faithful}, Ferdowsi et~al.\ 
developed thresholded hybrid models for 2-input CMOS \NOR\ and \NAND\ gates. These models 
are all based on replacing transistors by (possibly time-varying) switched resistors. All MIS
effects are covered by the model proposed in \cite{ferdowsi2023accurate, ferdowsi2024faithful},
which is a 4-state hybrid first-order model based on the Shichman-Hodges transistor model~\cite{ShichmanHodges}. 
Whereas the ODEs governing its 4 modes are all first-order, they have time-varying coefficients and are 
hence not trivial to solve analytically. The delay formulas derived in \cite{ferdowsi2023accurate} were 
hence based on piecewise approximations and thus quite complex. In \cite{ferdowsi2024faithful}, however, 
explicit solutions for the ODEs were determined and used for deriving simple and accurate analytic 
delay formulas. Thanks to a simple and fast parametrization procedure, only three delay
values ($\delta(0)$, $\delta(\infty)$ and $\delta(-\infty)$ are needed for computing the model
parameters for a given gate. Using a comparison to SPICE-generated traces, the authors showed
that appropriately parametrized versions of their model predict the actual delay of \NOR\ gates 
implemented in different CMOS technologies accurately and very fast, for any value of the input
separation time $\Delta$.

A serious limitation of the above delay models is that they only consider 
gates in isolation, however, i.e., without any interconnect. Unfortunately,
wires do have a substantial effect on circuit delays in practice:
they have non-negligible parasitic capacitances, resistances, and inductances,
which are spatially distributed and hence change with the wire length. 
The present paper will introduce an interconnect-extended version of the 
model of \cite{ferdowsi2024faithful}, presented in \cref{sec:background},
and evaluate its accuracy.

Compared to the existing MIS models surveyed above, our interconnect-augmented
model will be
faithful w.r.t.\ the SPF problem, cp.~\cite{FNS16:ToC}, as simple,
fast and efficient w.r.t.\ gate characterization as \cite{SRC15:TODAES},
will provide analytical and efficiently computable delay formulas that facilitate
very fast dynamic timing analysis, and last but not least will offer a good
delay prediction accuracy.

\section{Starting Point: A Thresholded Hybrid Model for \NOR\ Gates}
\label{sec:background}

In this section, we briefly revisit the cornerstones of the advanced thresholded hybrid model for \NOR\ (as well as \NAND) gates
introduced in \cite{ferdowsi2023accurate,ferdowsi2024faithful}. In \cref{Sec:IntModel}, we will extend this model
to also incorporate the interconnect between gates.

\begin{figure}[h]
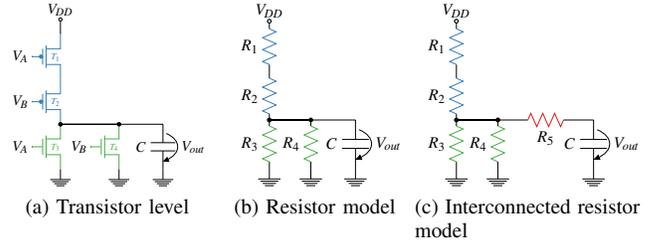

  \centering
  \subfloat[Transistor level]{
\includegraphics[height=0.28\linewidth]{\figPath{nor_RC_colored.pdf}}  
    \label{fig:nor_CMOS}}
  \hfil
  \subfloat[Resistor model]{
 \includegraphics[height=0.28\linewidth]{\figPath{nor_R_colored.pdf}}%
    \label{FigureNOR-GATE}}
  \hfil
  \subfloat[Interconnected resistor model]{
 \includegraphics[height=0.28\linewidth]{\figPath{nor_R_interconnect2.pdf}}%
    \label{FigureNOR-GATE_Int}}    
  \caption{Transistor schematic and the resistor model of a CMOS \NOR\ gate along with its augmented RC interconnect component.}
\end{figure}

Rather than representing all transistors by zero-time switches as in \cite{FMOS22:DATE}, the model of a 2-input CMOS \NOR\ gate proposed in \cite{ferdowsi2023accurate,ferdowsi2024faithful} replaces (some) transistors in \cref{fig:nor_CMOS} by time-varying resistors, as shown in \cref{FigureNOR-GATE}. These resistors, denoted as $R_i(t)$ for $i \in \{1,\ldots,4 \}$, vary between a predetermined on-resistance $R_i$ and the off-resistance $\infty$.
The governing principle for $R_i(t)$, which will be based on the Shichman-Hodges transistor model~\cite{ShichmanHodges}, is contingent upon the state of the corresponding input signal at the gate of the corresponding transistor.

This results in a hybrid model comprising four distinct modes, according to the four possible input states $(A,B)\in \{(0,0), (0,1), (1,0), (1,1) \}$. \cref{tab:T1} shows all possible input state transitions and the corresponding
resistor time evolution mode switches. Double arrows in the mode switch names
indicate MIS-relevant modes, whereas $+$ and $-$ indicate whether input $A$ switched
before $B$ or the other way round.  For instance, assume the system is in state
$(0,0)$ initially, i.e., that both $A$ and $B$ were set to 0 at time
$t_A= t_B= -\infty$. This causes $R_1$ and $R_2$ to be in the \emph{on-mode},
whereas $R_3$ and $R_4$ are in the \emph{off-mode}. If $A$ is switched to 1 at time
$t_A=0$, $R_1$ resp.\ $R_3$ is switched to the
off-mode resp.\ on-mode at time
$t^{\off}_1 = t^{\on}_3 = t_A = 0$. The corresponding mode switch is
$T_{-}^{\uparrow}$ and reaches state $(1,0)$. If $B$ is also
switched to 1 at some time $t_B=\Delta>0$, $R_2$ resp.\ $R_4$ is switched to the off-mode resp.\
on-mode at time $t^{\off}_2= t^{\on}_4 = t_B=\Delta$. The corresponding
mode switch is $T_{+}^{\uparrow\uparrow}$ and reaches state $(1,1)$. 

\begin{table}[t]
\centering
\caption{\small State transitions and modes. $\uparrow$ and $\uparrow \uparrow$ (resp.\ $\downarrow$ and $\downarrow \downarrow$) represent the first and the second rising (resp.\ falling) input transitions. $+$ and $-$ specify the sign of the switching time difference $\Delta=t_B-t_A$.}
\scalebox{0.65}
{
\begin{tabular}{lllllllllll}
\hline
Mode                            &  & Transition                &  & $t_A$       & $t_B$       &  & $R_1$                & $R_2$                & $R_3$                & $R_4$                \\ \cline{1-1} \cline{3-3} \cline{5-6} \cline{8-11} 
$T^{\uparrow}_{-}$              &  & $(0,0) \rightarrow (1,0)$ &  & $0$         & $-\infty$   &  & $on \rightarrow off$ & $on$                 & $off \rightarrow on$ & $off$                \\
$T^{\uparrow \uparrow}_{+}$     &  & $(1,0) \rightarrow (1,1)$ &  & $-|\Delta|$ & $0$         &  & $off$                & $on \rightarrow off$ & $on$                 & $off \rightarrow on$ \\
$T^{\uparrow}_{+}$              &  & $(0,0) \rightarrow (0,1)$ &  & $-\infty$   & $0$         &  & $on$                 & $on \rightarrow off$ & $off$                & $off \rightarrow on$ \\
$T^{\uparrow \uparrow}_{-}$     &  & $(0,1) \rightarrow (1,1)$ &  & $0$         & $-|\Delta|$ &  & $on \rightarrow off$ & $off$                & $off \rightarrow on$ & $on$                 \\
$T^{\downarrow}_{-}$            &  & $(1,1) \rightarrow (0,1)$ &  & $0$         & $-\infty$   &  & $off \rightarrow on$ & $off$                & $on \rightarrow off$ & $on$                 \\
$T^{\downarrow \downarrow}_{+}$ &  & $(0,1) \rightarrow (0,0)$ &  & $-|\Delta|$ & $0$         &  & $on$                 & $off \rightarrow on$ & $off$                & $on \rightarrow off$ \\
$T^{\downarrow}_{+}$            &  & $(1,1) \rightarrow (1,0)$ &  & $-\infty$   & $0$         &  & $off$                & $off \rightarrow on$ & $on$                 & $on \rightarrow off$ \\
$T^{\downarrow \downarrow}_{-}$ &  & $(1,0) \rightarrow (0,0)$ &  & $0$         & $-|\Delta|$ &  & $off \rightarrow on$ & $on$                 & $on \rightarrow off$ & $off$                \\ \hline
\end{tabular}}
\label{tab:T1}
\end{table}

A crucial part of the model is the choice of the governing principle dictating the temporal variation of $R_i(t)$ during the on- and off-mode: It should reasonably model the actual behavior of a transitor while facilitating the analytical solvability of the ensuing ordinary differential equation (ODE) systems. In 
\cite{ferdowsi2023accurate}, the continuous resistance model defined by
\begin{align}
R_j^{\on}(t) &= \frac{\alpha_j}{t-t^{\on}}+R_j; \ t \geq t^{\on}, \label{on_mode}\\
R_j^{\off}(t) &= \beta_j (t-t^{\off}) +R_j; \ t \geq t^{\off}, \label{off_mode}
\end{align}
for some constant slope parameters $\alpha_j$ [\si{\ohm\s}], $\beta_j$
[\si[per-mode=symbol]{\ohm\per\s}], and on-resistance $R_j$ [\si{\ohm}] is used;
$t^{\on}$ resp.\
$t^{\off}$ represent the time when the respective transistor is switched on
resp.\ off. These equations are based on the Shichman-Hodges transistor model~\cite{ShichmanHodges}, 
which assumes a quadratic correlation between the output current and the input voltage: \cref{on_mode}
and \cref{off_mode} follow from approximating the latter by 
$d \sqrt{t-t_0}$ in the operation range close to the threshold voltage $\vth$, 
with $d$ and $t_0$ denoting appropriate fitting parameters.
 
Interestingly, continuously varying resistors are only needed for switching on the \pmos\ transistors
in \cite{ferdowsi2023accurate}. In addition, rather than including $R_1(t)$ and $R_2(t)$ in the state of the
ODEs governing the appropriate modes, which would blow-up their dimensions, they are incorporated by means
of time-varying coefficients in simple first-order ODEs. All other transistor switchings, i.e., both the 
switching-off of the \pmos\ transistors
and any switching on or off of the \nmos\ transistors, happen instantaneously, as already employed in \cite{FMOS22:DATE}, which is accomplished by choosing the model parameters $\alpha_i=0$ in \cref{on_mode}for $i \in \{3,4 \}$,
and $\beta_k=\infty$ in \cref{off_mode} for $k \in \{1, \ldots ,4 \}$. 

In what follows, we will use
the notation $R_1 = R_{p_{A}}$, $R_2 = R_{p_{B}}$ with the abbreviation $2R = R_{p_{A}} + R_{p_{B}}$ for the two \pmos\ transistors $T_1$ and $T_2$ and $R_3 = R_{n_{A}}$, $R_4 = R_{n_{B}}$ for the two \nmos\ transistors $T_3$ and $T_4$.
Applying Kirchhoff's rules to \cref{FigureNOR-GATE} results
in the the non-autonomous, non-homogeneous ODE with non-constant coefficients

\begin{equation}
\label{Eq:ODE_base}
\frac{\dd V_{out}}{\dd t} =   -\frac{V_{out}}{C\,R_g(t)}+U(t),
\end{equation}
where
$\frac{1}{R_g(t)}=\frac{1}{R_1(t)+R_2(t)}+\frac{1}{R_3(t)}+\frac{1}{R_4(t)}$ and
$U(t)=\frac{V_{DD}}{C(R_1(t)+R_2(t))}$. It is well-known that the general solution of \eqref{Eq:ODE_base} is
\begin{equation}
\label{Eq1}
    V_{out}(t)= V_0\ e^{-G(t)} + \int_{0}^{t} U(s)\ e^{G(s)-G(t)}\dd s,
\end{equation}
where $V_0=V_{out}(0)$ denotes the initial condition and $G(t) = \int_{0}^{t}
(C\,R_g(s))^{-1} \dd s$.
As comprehensively described in \cite{ferdowsi2023accurate}, depending on each particular resistor's mode in each input state transition, different expressions for $R_g(t)$ and $U(t)$ are obtained. Denoting $I_1= \int_{0}^{t} \frac{\dd s}{R_1(s)+R_2(s)}$, $I_2= \int_{0}^{t}\frac{\dd s}{R_3(s)}$, and $I_3=\int_{0}^{t} \frac{\dd s}{R_4(s)}$, \cref{T:InerInt} summarizes those. The following \cref{thm:MISOuttraj} provides the resulting analytic formulas.

\begin{table}[h]
\centering
\caption{Integrals $I_1(t)$, $I_2(t)$, $I_3(t)$ and the function $U(t)$ for every possible mode switch; $\Delta=t_B-t_A$, and $2R=R_{p_A}+R_{p_B}$.}
\scalebox{0.58}
{
\begin{tabular}{llllll}
\hline
Mode                            &  & $I_1(t)= \int_{0}^{t} \frac{\dd s}{R_1(s)+R_2(s)}$                         & $I_2(t)= \int_{0}^{t}\frac{\dd s}{R_3(s)}$ & $I_3(t)=\int_{0}^{t} \frac{\dd s}{R_4(s)}$ & $U(t)= \frac{\vdd}{C(R_1(t)+R_2(t))}$                                                               \\ \cline{1-1} \cline{3-6} 
$T^{\uparrow}_{-}$              &  & $0$                                                                        & $\int_{0}^{t} (1/R_{n_A})\dd s$            & $0$                                        & $0$                                                                                                 \\
$T^{\uparrow \uparrow}_{+}$     &  & $0$                                                                        & $\int_{0}^{t} (1/R_{n_A})\dd s$            & $\int_{0}^{t} (1/R_{n_B}) \dd s$           & $0$                                                                                                 \\
$T^{\uparrow}_{+}$              &  & $0$                                                                        & $0$                                        & $\int_{0}^{t} (1/(R_{n_B}) \dd s$          & $0$                                                                                                 \\
$T^{\uparrow \uparrow}_{-}$     &  & $0$                                                                        & $\int_{0}^{t} (1/R_{n_A}) \dd s$           & $\int_{0}^{t} (1/R_{n_B}) \dd s$           & $0$                                                                                                 \\
$T^{\downarrow}_{-}$            &  & $0$                                                                        & $0$                                        & $\int_{0}^{t} (1/R_{n_B}) \dd s$           & $0$                                                                                                 \\
$T^{\downarrow \downarrow}_{+}$ &  & $\int_{0}^{t}(1/(\frac{\alpha_1}{s+\Delta}+\frac{\alpha_2}{s}+2R))\dd s$   & $0$                                        & $0$                                        & $\frac{\vdd t(t+ \Delta)}{C(2 R t^2 +(\alpha_1 + \alpha_2 + 2 \Delta R)t + \alpha_2 \Delta)}$       \\
$T^{\downarrow}_{+}$            &  & $0$                                                                        & $\int_{0}^{t} (1/(R_{n_A}) \dd s$          & $0$                                        & $0$                                                                                                 \\
$T^{\downarrow \downarrow}_{-}$ &  & $\int_{0}^{t}(1/(\frac{\alpha_1}{s}+\frac{\alpha_2}{s+|\Delta|}+2R))\dd s$ & $0$                                        & $0$                                        & $\frac{\vdd t(t+ |\Delta|)}{C(2 R t^2 +(\alpha_1 + \alpha_2 + 2 |\Delta| R)t + \alpha_1 |\Delta|)}$ \\ \hline
\end{tabular}}
\label{T:InerInt}
\end{table}

\begin{thm}[Output voltage trajectories for the \NOR\ gate {\cite[Theorems~6.2 and 6.3]{ferdowsi2024faithful}}] \label{thm:MISOuttraj}
For any $0 \leq |\Delta| \leq \infty$, the output voltage trajectory functions of our model for rising input transitions are given by
{\small
\begin{flalign}
V_{out}^{T^{\uparrow}_{-}}(t) &= V_{out}^{T^{\uparrow}_{-}}(0) e^{\frac{-t}{C R_{n_{A}}}},
\label{outsig1}\\
V_{out}^{T^{\uparrow}_{+}} (t) &= V_{out}^{T^{\uparrow}_{+}}(0) e^{\frac{-t}{C R_{n_{B}}}},
\label{outsig1neg}\\
V_{out}^{T^{\uparrow \uparrow}_{+}}(t) &=V_{out}^{T^{\uparrow}_{-}} (\Delta)  e^{- \bigl(\frac{1}{CR_{n_A}}+\frac{1}{CR_{n_B}}\bigr)t},
\label{outsig2}\\
V_{out}^{T^{\uparrow \uparrow}_{-}}(t) &=V_{out}^{T^{\uparrow}_{+}} (\Delta)  e^{- \bigl(\frac{1}{CR_{n_A}}+\frac{1}{CR_{n_B}}\bigr)t}.
\label{outsig2_neg}
\end{flalign}
}
The output voltage trajectory functions for falling input transitions are given by
{\small
\begin{flalign}
V_{out}^{T^{\downarrow}_{-}}(t) &= V_{out}^{T_{-}^{\downarrow}}(0) e^{\frac{-t}{CR_{n_B}}},
\label{eq:FirstFalltheorem}\\
V_{out}^{T^{\downarrow}_{+}}(t) &= V_{out}^{T_{+}^{\downarrow}}(0) e^{\frac{-t}{CR_{n_A}}},
\label{eq:FirstFalltheorem_plus}\\
V_{out}^{T^{\downarrow \downarrow}_{+}}(t)&= \vdd  + \bigl(V_{out}^{T^{\downarrow}_{-}}(\Delta) -\vdd   \bigr) \label{SoughtOutput} \\
 & \cdot\left[ e^{\frac{-t}{2RC}} \Bigl(1+\frac{2t}{d+\sqrt{\chi}}\Bigr)^{\frac{-A+a}{2RC}} \Bigl(1+\frac{2t}{d-\sqrt{\chi}}\Bigr)^{\frac{A}{2RC}} \right],\nonumber
\end{flalign}
}
where $a=\frac{\alpha_1+\alpha_2}{2R}$, $d=a+\Delta$,  $ \chi=d^2-4c'$, $c'=\frac{\alpha_2 \Delta}{2R}$, and $A=\frac{\alpha_2\Delta - aR(d- \sqrt{\chi})}{2R\sqrt{\chi}}$. The output voltage trajectory $V_{out}^{T^{\downarrow \downarrow}_{-}}(t)$ for
negative $\Delta$ is
obtained from \cref{SoughtOutput} by exchanging $\alpha_1$ and $\alpha_2$,
$V_{out}^{T^{\downarrow}_{-}}(\Delta)$ by $V_{out}^{T^{\downarrow}_{+}}(|\Delta|)$, and $\Delta$ by $|\Delta|$ in $d$, $\chi$ and $A$.
\end{thm}

The trajectories $V_{out}^{T^{\uparrow \uparrow}_{+}}(t)$, $V_{out}^{T^{\uparrow \uparrow}_{-}}(t)$
and $V_{out}^{T^{\downarrow \downarrow}_{+}}(t)$, $V_{out}^{T^{\downarrow \downarrow}_{-}}(t)$ in
\cref{thm:MISOuttraj} are specifically tailored to facilitate the computation of
the MIS delays.
For example, the formulas \cref{outsig2} and \cref{SoughtOutput}
(for $\Delta \geq 0$) have been determined according to the following two procedures:
\begin{itemize}
\item[(i)] Compute $V_{out}^{T^{\uparrow}_{-}} (\Delta)$ for the first transition $(0,0)\to(1,0)$
  and use it as the initial value for $V_{out}^{T^{\uparrow \uparrow}_{+}}(t)$ governing the second
  transition $(1,0)\to(1,1)$; the ultimately sought MIS delay $\delta_{M,+}^{\downarrow}(\Delta)$ is
  the time (measured from the first transition until either the first or the second trajectory
  crosses the threshold voltage $\vdd/2$ from
  above when the first one starts from $V_{out}^{T^{\uparrow}_{-}} (0)=\vdd$.
\item[(ii)] Compute $V_{out}^{T^{\downarrow}_{-}}(\Delta)$ for the first transition $(1,1)\to(0,1)$
  and use it as the initial value for
  $V_{out}^{T^{\downarrow \downarrow}_{+}}(t)$ governing the second transition $(0,1)\to(0,0)$;
  the ultimately sought MIS delay $\delta_{M,+}^{\uparrow}(\Delta)$
  is the time (measured from the second transition) until the second trajectory
  crosses the threshold voltage $\vdd/2$ from below when the first one starts from
  $V_{out}^{T^{\downarrow}_{-}}(0)=0$.
\end{itemize}

Clearly, actually determining these MIS delays requires inverting the appropriate
trajectory formulas, which turned out to be easy in the rising
input transition case (i) but difficult in the falling input transition case (ii). In \cite{ferdowsi2023accurate}, a
somewhat complicated piecewise approximation (in terms of $\Delta$) of both the trajectory and, hence, the 
corresponding delay formula was used (these approximations also formed the basis for the preliminary version
\cite{FFSS23:DSD} of the present paper). In \cite{ferdowsi2024faithful}, however, an explicit 
trajectory formula was found also for the falling input transition case. According to
\cref{SoughtOutput}, $\delta_{M,+}^{\uparrow}(\Delta)$ is the solution (in $t$) of the
implicit function $I(t,\Delta)=0$, where

{
\footnotesize
\begin{align}
I(t,\Delta)=e^{\frac{-t}{2RC}} \Bigl(1+\frac{2t}{d+\sqrt{\chi}}\Bigr)^{\frac{-A+a}{2RC}} \Bigl(1+\frac{2t}{d-\sqrt{\chi}}\Bigr)^{\frac{A}{2RC}} -\frac{1}{2}.
\label{BiVarFunction}
\end{align}
}

Unfortunately, since $\lim_{\Delta\to 0} A =0$ and $\lim_{\Delta\to 0} (d-\sqrt{\chi}) =0$, the point $(t,\Delta)=(0,0)$ is
singular for $I(t,\Delta)=0$, so solving the latter for $t=\delta(\Delta)$ by means of the implicit function theorem is impossible.
However, the bootstrapping method \cite{deB70} was successfully employed for developing an accurate asymptotic expansion of $\delta(\Delta)$ for $\Delta \to 0$. 
\cref{thm:MISdelay} provide the resulting MIS delay formulas for an isolated two-input \NOR\ gate,
for both rising and falling input transitions as well as positive and negative $\Delta$.

\begin{thm}[MIS delay functions for the \NOR\ gate {\cite[Theorems~6.4 and 6.5]{ferdowsi2024faithful}}] \label{thm:MISdelay}
  For any $0 \leq |\Delta| \leq \infty$, the MIS delay functions of our model for the rising and falling input transitions are respectively given by
  
{\footnotesize
\begin{align}
& \delta_{M,+}^{\downarrow}(\Delta) =  \nonumber \\ 
& \begin {cases}
 \frac{\log(2)CR_{n_A}R_{n_B} - \Delta R_{n_B}}{R_{n_A}+R_{n_B}} + \Delta + \delta_{min} &   \ \ 0 \leq \Delta < \log(2)CR_{n_A} \\ 
 \log(2)CR_{n_A}  + \delta_{min} &   \ \ \Delta \geq \log(2)CR_{n_A}
\end {cases} \label{FallingmisdelayformulaNOR_pos}
\end{align}
\begin{align}
&\delta_{M,-}^{\downarrow}(\Delta) = \nonumber \\ 
& \begin{cases}
 \frac{\log(2)CR_{n_A}R_{n_B} + |\Delta| R_{n_A}}{R_{n_A}+R_{n_B}} + |\Delta| + \delta_{min} &   \ \ |\Delta| < \log(2)CR_{n_B} \\ 
\log(2)CR_{n_B}  + \delta_{min} &   \ \ |\Delta| \geq \log(2)CR_{n_B}
\end {cases} \label{FallingmisdelayformulaNOR_neg}
\end{align}
}
{\footnotesize
\begin{align}
\delta_{M,+}^{\uparrow}(\Delta) &= \begin {cases}
\delta_{0} - \frac{\alpha_1}{\alpha_1+\alpha_2} \Delta + \delta_{min}  &   \ \ 0 \leq \Delta < \frac{(\alpha_1+\alpha_2)(\delta_{0} - \delta_{\infty})}{\alpha_1}   \\ 
\delta_{\infty} + \delta_{min} &   \ \ \Delta \geq \frac{(\alpha_1+\alpha_2)(\delta_{0} - \delta_{\infty})}{\alpha_1}
\end {cases}\label{Risingmisdelayformula}
\\
\delta_{M,-}^{\uparrow}(\Delta) &= \begin {cases}
\delta_{0} - \frac{\alpha_2}{\alpha_1+\alpha_2} |\Delta| + \delta_{min} &   \ \ 0 \leq |\Delta| < \frac{(\alpha_1+\alpha_2)(\delta_{0} - \delta_{-\infty})}{\alpha_2}   \\ 
\delta_{-\infty} + \delta_{min} &   \ \ |\Delta| \geq \frac{(\alpha_1+\alpha_2)(\delta_{0} - \delta_{-\infty})}{\alpha_2}
\end {cases}\label{Risingmisdelayformulaminus}
\end{align}
}
where
{\footnotesize
\begin{align}
\delta_{0} &= - \frac{\alpha_1 + \alpha_2}{2R} \Bigl[ 1+ W_{-1}\Bigl(\frac{-1}{e \cdot 2^{\frac{4R^2C}{\alpha_1+ \alpha_2}}}\Bigr) \Bigr],  \label{eq:delta0} \\
\delta_{\infty}&= -\frac{\alpha_2}{2R} \Bigl[ 1+ W_{-1}\Bigl(\frac{-1}{e \cdot 2^{\frac{4R^2C}{\alpha_2}}}\Bigr) \Bigr],  \label{eq:deltainf} \\
\delta_{-\infty}&= -\frac{\alpha_1}{2R} \Bigl[ 1+ W_{-1}\Bigl(\frac{-1}{e \cdot 2^{\frac{4R^2C}{\alpha_1}}}\Bigr) \Bigr]. \label{eq:deltaminf}
\end{align}
}
Herein, $y=W_{-1}(x)$ is the non-principal real branch of the Lambert $W$ function (that solves $ye^y=x$ for $y \leq -1$).
\end{thm}

\section{An Accurate Hybrid MIS Delay Model for Interconnected \NOR\ Gates}
\label{Sec:IntModel}

In this section, we will extend the thresholded hybrid model for an isolated \NOR\ gate surveyed in \cref{sec:background}
by adding a simple interconnect.
State-of-the-art interconnect modeling usually breaks up wires into segments,
each of which is characterized by some lumped model,
typically of $\Pi$, T, and RC type \cite{jan2003digital}, as depicted in
\cref{fig:Interconnect Models}.

\begin{figure}[h]
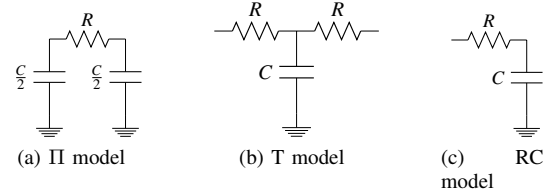

  \centering
  \subfloat[$\Pi$ model]{
    \includegraphics[width=0.20\linewidth]{\figPath{InterconnectModelgeneratorP.pdf}}%
    \label{fig:Interconnect ModelsPi}}
  \hfil
  \subfloat[T model]{
    \includegraphics[width=0.25\linewidth]{\figPath{InterconnectModelgeneratorT.pdf}}%
    \label{fig:Interconnect ModelsT}}
      \hfil
  \subfloat[RC model]{
    \includegraphics[width=0.14\linewidth]{\figPath{InterconnectModelgeneratorRC.pdf}}%
    \label{fig:Interconnect ModelsRC}}
  \caption{\small\em Lumped models for wires.}\label{fig:Interconnect Models}
\end{figure}

Since our focus is on exploring the modeling accuracy achievable with
first-order hybrid models, to preserve analytic solvability, we restrict our attention to the
lumped RC model,
as shown in \cref{fig:Interconnect ModelsRC}: Although this model is known to be less accurate than e.g.\
the $\Pi$ model in static timing analysis, it is the only one
that can be added to the gate model of \cite{ferdowsi2023accurate} without considerably increasing the complexity and turning
it into a second-order model: Adding a $\Pi$ model would add another
state-holding stage (capacitor) and hence raise the dimension of the
ODE systems to 2.

\subsection{Interconnect extension}

Integrating \cref{fig:Interconnect ModelsRC} into the \NOR\ gate results in the schematics 
shown in \cref{FigureNOR-GATE_Int}. By applying Kirchhoff's rules, the non-homogeneous ODE 

{\footnotesize
\begin{equation}
\label{EqIC0}
\frac{\dd V_{out}(t)}{\dd t}=-\frac{V_{out}(t)}{C\,R_g(t)(\frac{R_5}{R_g(t)}+1)}+\frac{\vdd}{C(R_1(t)+R_2(t))(\frac{R_5}{R_g(t)}+1)},
\end{equation}
}
is easily derived. Note that \cref{EqIC0} is just \cref{Eq:ODE_base} with the additional factor $f(t)=1/(\frac{R_5}{R_g(t)}+1)$ in all terms except $\frac{\dd V_{out}(t)}{\dd t}$. Since this non-constant factor $f(t)$ makes solving this ODE explicitly
very hard, if at all possible, we decided to take the ``easy route'' of approximating $f(t)$ by a constant value $F$: This allows us to re-use the solutions obtained for \cref{Eq:ODE_base}, by just replacing $C$ with $\frac{C}{F}$ in all output voltage trajectory formulas. 
In order to reduce the approximation error, however, we use different values of $F$ in different modes.

Recall that, in the original hybrid model, each mode switch enables some specific ODE system, the solution of which gives the respective trajectory. Fortunately, as can be observed in
\cref{tab:T1},
all transitions except $(0,1) \rightarrow (0,0)$ and $(1,0) \rightarrow (0,0)$ lead to a constant value for $R_g(t)$ and hence $F$ a priori. Consequently, for those six transitions, we can safely substitute $f(t)$ by the appropriate constant value. 

Unfortunately, this is not the case for the transitions $(0,1) \rightarrow (0,0)$ and $(1,0) \rightarrow (0,0)$, though, so replacing $f(t)$ by some constant value introduces some approximation error.
Fortunately, the time span during which $R_g(t)$ varies significantly here is very small. Moreover, its variability is not very large either: In particular, as the switch-on of a transistor is fast, one may reasonably conjecture that replacing $f(t)$ by $1/(\frac{R_5}{R_{g_{min}}}+1)$ should lead to a good approximation; and indeed, the results
of our validation experiments in \cref{Sec:Result} will confirm this conjecture. The fact that $1/R_{g_{min}} = 1/(2R)$ follows from
\cref{tab:T1}, which reveals that $(0,1) \rightarrow (0,0)$ resp. $(1,0) \rightarrow (0,0)$ leads to $1/R_g(t)=1/(\frac{\alpha_1}{t+ \Delta}+ \frac{\alpha_2}{t}+2R)$ resp. $=1/(\frac{\alpha_1}{t}+ \frac{\alpha_2}{t+ \Delta}+2R)$. 
\cref{table:IC} summarizes all exact and approximate values of $R_g(t)$ and $F$ corresponding to each mode switch.

\begin{table}[t]
\centering
\caption{\small\em Input mode switching and the resulting values for $R_g(t)$ and the corresponding approximation $F$ for $f(t)$.}
\label{table:IC}
\scalebox{0.75}{
\begin{tabular}{ccccc}
\hline
MS                        &  & $R_g(t)$                                   &  & $F$                                                        \\ \cline{1-1} \cline{3-3} \cline{5-5} 
$(0,0) \rightarrow (1,0)$ and $(1,1) \rightarrow (1,0)$ &  & $=R_{n_A}$                                 &  & $=\frac{R_{n_A}}{R_5 + R_{n_A}}$                                \\
$(1,0) \rightarrow (1,1)$ and $(0,1) \rightarrow (1,1)$ &  & $=\frac{R_{n_A}R_{n_B}}{R_{n_A}+ R_{n_B}}$ &  & $= \frac{R_{n_A}R_{n_B}}{R_5(R_{n_A}+R_{n_B})+ R_{n_A}R_{n_B}}$ \\
$(0,0) \rightarrow (0,1)$ and $(1,1) \rightarrow (0,1)$ &  & $=R_{n_B}$                                 &  & $=\frac{R_{n_B}}{R_5 + R_{n_B}}$                                \\
$(0,1) \rightarrow (0,0)$ and $(1,0) \rightarrow (0,0)$&  & $\approx 2R$                               &  & $\approx \frac{2R}{R_5+ 2R}$                                    \\ \hline
\end{tabular}}
\end{table} 

The results of the above discussion lead to the following \cref{Corol:1},
which gives the delay predictions of our interconnect-augmented model. It is identical
to \cref{thm:MISdelay}, except that $C$ is replaced by $C/F$ for the appropriate
transitions (as determined by the procedure (i) and (ii) in \cref{sec:background}), 
where $F$ is given in \cref{table:IC}.

\begin{corol}[MIS delay functions for the interconnect-augmented \NOR\ gate]\label{Corol:1}
For any $0 \leq |\Delta| \leq \infty$, the MIS delay functions of our interconnect-augmented model for rising and falling input transitions are respectively given by

{\footnotesize
\begin{align}
& \delta_{M,+}^{\downarrow}(\Delta) =  \nonumber \\ 
& \begin {cases} 
 \frac{\log(2)C_2R_{n_A}R_{n_B} - \frac{C_2}{C_1} \Delta R_{n_B}}{R_{n_A}+R_{n_B}} + \Delta + \delta_{min} &   \ \ 0 \leq \Delta < \log(2)C_1R_{n_A}  \\ 
 \log(2)C_1R_{n_A} + \delta_{min} &   \ \ \Delta \geq \log(2)C_1R_{n_A}
\end {cases} \label{Fallingmisdelayformula_int}
\end{align}
\begin{align}
& \delta_{M,-}^{\downarrow}(\Delta) = \nonumber \\ 
&  \begin{cases}
 \frac{\log(2)C_2R_{n_A}R_{n_B} + \frac{C_2}{C_1'} |\Delta| R_{n_A}}{R_{n_A}+R_{n_B}} + |\Delta| + \delta_{min} &   \ \ |\Delta| < \log(2)C_1^{'}R_{n_B}  \\ 
\log(2)C_1^{'}R_{n_B} + \delta_{min} &   \ \ |\Delta| \geq \log(2)C_1^{'}R_{n_B}
\end {cases} \label{Fallingmisdelayformulaminus_int}
\end{align}
}
{\footnotesize
\begin{align}
\delta_{M,+}^{\uparrow}(\Delta) &= \begin {cases}
\delta_{0} - \frac{\alpha_1}{\alpha_1+\alpha_2} \Delta + \delta_{min}  &   \ \ 0 \leq \Delta < \frac{(\alpha_1+\alpha_2)(\delta_{0} - \delta_{\infty})}{\alpha_1}   \\ 
\delta_{\infty} + \delta_{min} &   \ \ \Delta \geq \frac{(\alpha_1+\alpha_2)(\delta_{0} - \delta_{\infty})}{\alpha_1}
\end {cases}\label{Risingmisdelayformula_int}
\end{align}
\begin{align}
\delta_{M,-}^{\uparrow}(\Delta) &= \begin {cases}
\delta_{0} - \frac{\alpha_2}{\alpha_1+\alpha_2} |\Delta| + \delta_{min} &   \ \ 0 \leq |\Delta| < \frac{(\alpha_1+\alpha_2)(\delta_{0} - \delta_{-\infty})}{\alpha_2} \\ 
\delta_{-\infty} + \delta_{min} &   \ \ |\Delta| \geq \frac{(\alpha_1+\alpha_2)(\delta_{0} - \delta_{-\infty})}{\alpha_2}
\end {cases}\label{Risingmisdelayformulaminus_int}
\end{align}
}
where
{\small
\begin{align}
\delta_{0} &= - \frac{\alpha_1 + \alpha_2}{2R} \Bigl[ 1+ W_{-1}\Bigl(\frac{-1}{e \cdot 2^{\frac{4R^2C_3}{\alpha_1+ \alpha_2}}}\Bigr) \Bigr],  \label{eq:delta0_int}
\end{align}
\begin{align}
\delta_{\infty}&= -\frac{\alpha_2}{2R} \Bigl[ 1+ W_{-1}\Bigl(\frac{-1}{e \cdot 2^{\frac{4R^2C_3}{\alpha_2}}}\Bigr) \Bigr],  \label{eq:deltainf_int}
\end{align}
\begin{align}
\delta_{-\infty}&= -\frac{\alpha_1}{2R} \Bigl[ 1+ W_{-1}\Bigl(\frac{-1}{e \cdot 2^{\frac{4R^2C_3}{\alpha_1}}}\Bigr) \Bigr], \label{eq:deltaminf_int}
\end{align}
\begin{align}
C_1&=\frac{C(R_5+R_{n_{A}})}{R_{n_{A}}}, \label{eq:C_1_int}
\end{align}
\begin{align}
C_1^{'} &=\frac{C(R_5+R_{n_{B}})}{R_{n_{B}}}, \label{eq:C_1_prime_int}
\end{align}
\begin{align}
C_2 &=\frac{C\bigl(R_5(R_{n_{A}}+R_{n_{B}})+R_{n_{A}}R_{n_{B}}\bigl)}{R_{n_{A}}R_{n_{B}}}, \label{eq:C_2_int}
\end{align}
\begin{align}
C_3&=\frac{C(R_5+2R)}{2R}. \label{eq:C_3_int}
\end{align}
}
\end{corol}
\begin{proof}

We give the proof for $\delta_{M,+}^{\downarrow}(\Delta)$ only; the expression for $\delta_{M,-}^{\downarrow}(\Delta)$ can be derived from the former by exchanging $R_{n_A}$ and $R_{n_B}$ and replacing $\Delta$ by $|\Delta|$. First, consider the transition $(1,0) \rightarrow (1,1)$ and the corresponding trajectory $V_{out}^{T^{\uparrow \uparrow}{+}}(t)$ in procedure (i) stated after \cref{thm:MISOuttraj}. According to \cref{table:IC}, it is just \cref{outsig2} with $C$ replaced by $C_2$ here, i.e.,
\begin{equation}
  V_{out}^{T^{\uparrow \uparrow}_{+}}(t)= V_{out}^{T^{\uparrow}_{-}} (\Delta)  e^{- \bigl(\frac{1}{C_2R_{n_A}}+\frac{1}{C_2R_{n_B}}\bigr)t}.\label{outsig2withC2}
\end{equation}
This trajectory must start from the initial value
\begin{equation}
  V_{out}^{T^{\uparrow}_{-}}(\Delta) = V_{out}^{T^{\uparrow}_{-}}(0) e^{\frac{-\Delta}{C_1 R_{n_{A}}}},\label{outsig1withC1}
\end{equation}
  associated with the transition $(0,0) \rightarrow (1,0)$, which results from replacing $C$ by $C_1$ in \cref{outsig1}. The latter, in turn, starts from $V_{out}^{T^{\uparrow}_{-}}(0)=\vdd$. The goal is to determine the time $\delta_{M,+}^\downarrow(\Delta)$ when $\vdd/2$ is reached from above either by (a) the first trajectory $V_{out}^{T^{\uparrow}_{-}}(t)$ if $\Delta$ is large, or (b) by $V_{out}^{T^{\uparrow \uparrow}_{+}}(t)$ itself (which commences at time $\Delta$, i.e., $t=0$ corresponds to $\Delta$ here) if $\Delta$ is small enough. Given that both trajectories involve only a single exponential function, they are straightforward to invert. From \cref{outsig1withC1}, it is evident that case (a) occurs for $\Delta \geq \log(2)C_1R_{n_A}$, while \cref{outsig2withC2} governs case (b) for smaller values of $\Delta$. It is not hard to confirm that this gives raise
  to \cref{Fallingmisdelayformula_int}, which differs from \cref{FallingmisdelayformulaNOR_pos} only in that
  $\Delta R_{n_B}$ in the numerator for case (b) has been replaced by $\frac{C_2}{C_1}\Delta R_{n_B}$. 

  Obtaining $\delta_{M,+}^{\uparrow}(\Delta)$ is even simpler, since $V_{out}^{T^{\downarrow}_{-}}(\Delta)=0$ in
  procedure (ii), as it starts from $V_{out}^{T^{\downarrow}_{-}}(0)=0$ and follows \cref{eq:FirstFalltheorem}.
  Consequently, only $V_{out}^{T^{\downarrow \downarrow}_{+}}(t)$ starting from initial value 0 is relevant here.
  All that is needed is hence to replace $C$ by $C_3$ in \cref{SoughtOutput}. Consequently, the
  MIS delay formula \cref{Risingmisdelayformula} remains valid,
  provided $C$ is replaced by $C_3$ in \cref{eq:delta0}-\cref{eq:deltaminf}. This justifies
  \cref{Risingmisdelayformula_int} and \cref{eq:delta0_int}-\cref{eq:deltaminf_int}. The MIS
  delay formula \cref{Risingmisdelayformulaminus_int} for negative $\Delta$ follows from exchanging
  $\alpha_1$ and $\alpha_2$ and replacing $\Delta$ by $|\Delta|$ in \cref{Risingmisdelayformula_int}.
\end{proof}

\subsection{Model parametrization}
\label{Sec:Param}
For the applicability of \cref{Corol:1}, it is essential to have a practical procedure for model parameterization:
Given the \emph{extremal} MIS delay values of a real interconnected \NOR\ gates, namely $\ddoD_S(-\infty)$, $\ddoD_S(0)$, and $\ddoD_S(\infty)$ 
according to \cref{corFig3}, and $\dupD_S(-\infty)$, $\dupD_S(0)$, and $\dupD_S(\infty)$ according to \cref{corFig5},
we want to determine suitable values for the parameters $\alpha_1$, $\alpha_2$, $C$, $R$, $R_{n_A}$, $R_{n_B}$, and $R_5$
such that the MIS delays predicted by our model \emph{match} these values, in the sense that 
$\ddoD_{M,-}(-\infty)=\ddoD_S(-\infty)$, $\ddoD_{M,-}(0)=\ddoD_{M,+}(0)=\ddoD_S(0)$, $\ddoD_{M,+}(\infty)=\ddoD_S(\infty)$ and $\dupD_{M,-}(-\infty)=\dupD_S(-\infty)$,  $\dupD_{M,-}(0)=\dupD_{M,+}(0)=\dupD_S(0)$, $\dupD_{M,+}(\infty)=\dupD_S(\infty)$.

For the isolated \NOR\ gate model proposed in \cite{ferdowsi2023accurate}, matching parameter values could 
only be determined after adding some well-chosen minimal pure delay $\dmin>0$ to the model. Note carefully that
adding such a non-zero minimal pure delay to the model is also mandatory for making it \emph{causal}, see \cite{FFNS23:HSCC,ferdowsi2024faithful} for details. Thanks to the explicit delay formulas developed in 
\cite{ferdowsi2024faithful}, given in \cref{thm:MISdelay}, the least-squares fitting-based parametrization 
procedure used in \cite{ferdowsi2023accurate} could be replaced by explicit formulas for computing
the sought parameters. These formulas also allowed to compute the required value for $\dmin$, 
see \cite[Thm.~6.6]{ferdowsi2024faithful}. 

Since we cannot re-use these analytic parametrization formulas for our interconnect-augmented 
model directly, due to the additional parameter $R_5$, \cref{thm:gatechar} provides a suitably adapted parameterization procedure. Interestingly, the additional degree
of freedom provided by $R_5$ completely removed the need for a uniquely determined value of
the pure delay $\dmin$ as in \cite{ferdowsi2024faithful}. Indeed, it turned out that the new parametrization procedure works for almost any reasonable choice of $\dmin$, which we now use
for modeling an additional pure delay of the interconnect at the output. The model parameters
computed by our parametrization formulas below make sure that the given MIS delay values will
be matched. 
For determining the particular value of $\dmin$ for our validation in \cref{Sec:Result}, we used the procedure for measuring the pure delay of inverters (which, unlike $\ddoD_S(0)$ 
and $\dupD_S(0)$,
abstracts away the delay caused by the finite slope of the analog waveforms) proposed by 
Maier et~al.\ in \cite{maier2021composable}, by tying together inputs $A$ and $B$ of 
our \NOR\ gates.
It turned out, however, that the results are insensitive to the actual choices of $\dmin$.

\begin{thm}[Model parametrization for interconnect-augmented \NOR\ gates]\label{thm:gatechar}
Let $\dmin\geq 0$ be some interconnect pure delay and $\ddoD_S(-\infty)$,  $\ddoD_S(0)$, $\ddoD_S(\infty)$ and $\dupD_S(-\infty)$,  $\dupD_S(0)$, $\dupD_S(\infty)$ be the MIS delay values of a real interconnected \NOR\ gate that shall be matched by our model. Given an arbitrary chosen value $C$ for the load capacitance, this is accomplished by choosing the model parameters as follows:

\begin{flalign}
&R_5 =  \frac{\bigl( \ddoD_S(0) - \dmin - \epsilon \bigr)}{\log(2)C}   \label{eq:dmin}&
\end{flalign}
\begin{flalign}
&R_{n_{A}}=\frac{\ddoD_S(\infty)-\ddoD_S(0)+\epsilon}{\log(2) C} \label{eq:rna}&
\end{flalign}
\begin{flalign}
&R_{n_{B}}=\frac{\ddoD_S(-\infty)-\ddoD_S(0)+\epsilon}{ \log(2) C} \label{eq:rnb}&
\end{flalign}
\begin{flalign}
&\epsilon=\sqrt{\bigl(\ddoD_S(\infty)-\ddoD_S(0)\bigr)\bigl(\ddoD_S(-\infty)-\ddoD_S(0)\bigr)} &
\end{flalign}

Furthermore, using the function
{\scriptsize
\begin{flalign}
&A(t,R,R_5,C)=& \nonumber \\
&\frac{-2R \bigl(t-C(R_5+2R) \cdot \log(2) \bigr)}{W_{-1}\Bigl( \bigl(\frac{C(R_5+2R) \cdot \log(2)}{t}-1\bigr) e^{\frac{C(R_5+2R) \cdot \log(2)}{t}-1} \Bigl) +1 - \frac{C(R_5+2R) \cdot \log(2)}{t}}\label{eq:AtRC},&
\end{flalign}
}
determine $R$ by numerically\footnote{Whereas there might be a way to solve it analytically, we did not find it so far.} solving the equation 
{\small
\begin{flalign}
&A\bigl(\dupD_S(0)-\dmin,R,R_5,C\bigr)-A\bigl(\dupD_S(\infty)-\dmin,R,C\bigr)&\nonumber \\
& - A\bigl(\dupD_S(-\infty)-\dmin,R,C\bigr) = 0 \label{eq:forR}, &
\end{flalign}
}
and finally choose

\begin{flalign}
\alpha_1 &= A\bigl(\dupD_S(-\infty)-\dmin,R,R_5,C\bigr)\label{eq:alpha1},\\
\alpha_2 &= A\bigl(\dupD_S(\infty)-\dmin,R,R_5,C\bigr)\label{eq:alpha2}.
\end{flalign}
\end{thm}

\begin{proof}
The proof follows the general strategy of the proof of \cite[Thm.~6.6]{ferdowsi2024faithful}.
We first consider the parameters determined by the rising input transition 
case. 
To align the 
delay formulas in \cref{Corol:1} with the given extremal
delay values, we just plug in $\ddoD_S(-\infty)-\dmin$, $\ddoD_S(0)-\dmin$, and $\ddoD_S(\infty)-\dmin$ to obtain the following system of equations for our
sought parameters $\dmin$, $R_{n_{A}}$, $R_{n_{B}}$, and $R_5$:

{\small
\begin{flalign}
&\ddoD_S(0)-\dmin- \frac{\log(2) \cdot C_2  \cdot  R_{n_{A}} R_{n_{B}}}{R_{n_{A}}+R_{n_{B}}}=0 \label{Eq0:RnARnB} \\
& \ddoD_S(\infty)-\dmin- \log(2) \cdot C_1  \cdot R_{n_{A}}=0 \label{Eq0:RnA} \\
& \ddoD_S(-\infty)-\dmin- \log(2) \cdot C_1'  \cdot R_{n_{B}}=0 \label{Eq0:RnB}
\end{flalign}
Plugging \cref{eq:C_2_int} into \cref{Eq0:RnARnB} gives:
\begin{flalign}
&\ddoD_S(0)-\dmin = \frac{\log(2) \cdot C [R_5(R_{n_{A}}+R_{n_{B}}) + R_{n_{A}}R_{n_{B}}]}{R_{n_{A}}+R_{n_{B}}}, \nonumber
\end{flalign}
}
which leads to 
\begin{flalign}
&  \frac{\ddoD_S(0)-\dmin}{\log(2) C} = R_5 + \frac{R_{n_{A}}R_{n_{B}}}{R_{n_{A}}+R_{n_{B}}}, \nonumber
\end{flalign}
the reciprocal of which reads
\begin{flalign}
\frac{1}{R_{n_{A}}}+\frac{1}{R_{n_{B}}} &= \frac{\log(2) C}{\ddoD_S(0)-\dmin - \log(2)CR_5}. \label{Eq:RnARnB}
\end{flalign}
Following a similar approach for \cref{Eq0:RnA} and \cref{Eq0:RnB} and using \cref{eq:C_1_int} and \cref{eq:C_1_prime_int} leads to
\begin{flalign}
\frac{1}{R_{n_{A}}} &= \frac{\log(2) C}{\ddoD_S(\infty)-\dmin-\log(2)CR_5} \label{Eq:RnA}, \\
\frac{1}{R_{n_{B}}} &= \frac{\log(2) C}{\ddoD_S(-\infty)-\dmin-\log(2)CR_5} \label{Eq:RnB}.
\end{flalign}
Now, from \cref{Eq:RnARnB}, \cref{Eq:RnA}, and \cref{Eq:RnB}, it follows that
{\small
\begin{flalign}
&\frac{1}{\ddoD_S(0)-\dmin-\log(2)CR_5}= & \nonumber \\ 
&\frac{1}{\ddoD_S(\infty)-\dmin-\log(2)CR_5} + \frac{1}{\ddoD_S(-\infty)-\dmin-\log(2)CR_5}, \nonumber  &
\end{flalign}
}
which can be rewritten into a quadratic equation for $\dmin+\log(2)CR_5$. Choosing the
negative solution, which ensures that $\dmin+\log(2)CR_5 \leq \ddoD_S(0)$, provides
the expression for $R_5$ stated in \cref{eq:dmin}. Plugging it
into \cref{Eq:RnA} and \cref{Eq:RnB} gives us $R_{n_{A}}$ resp.\ $R_{n_{B}}$ according to \cref{eq:rna} resp.\ \cref{eq:rnb}.

Next, we focus on the parameters determined by the falling input transition case. We explain \cref{eq:AtRC} by considering $A\bigl(\dupD_S(0) - \dmin, R, R_5, C\bigr)$, which corresponds to setting $t = \delta_0 = \dupD_S(0) - \dmin$ as defined in \cref{eq:delta0_int}. We start out from the implicit function $I(t,\Delta)=0$ defined in
\cref{BiVarFunction} for $\Delta=0$ and $t=\delta_0$, which also causes $A=0$ and $\sqrt{\chi}=d=a=(\alpha_1+\alpha_2)/2R$. Recall that it is adapted to our interconnected \NOR\ gate setting just by replacing $C$ by $C_3$
given in \cref{eq:C_3_int}. We obtain
\begin{equation}
e^{-\frac{\delta_0}{2RC_3}}\Bigl(1+\frac{\delta_0}{a}\Bigr)^{\frac{a}{2RC_3}}=\frac{1}{2}.\label{eq:traj0first}
\end{equation}
Abbreviating $\alpha = \alpha_1+\alpha_2$ and noting $a=\alpha/2R$, raising \cref{eq:traj0first} 
to the power $2RC_3/a=C(R_5+2R)/a$ results in
\begin{equation}
e^{-\frac{\delta_0}{a}}\Bigl(1+\frac{\delta_0}{a}\Bigr)= 2^{-\frac{C(R_5+2R)}{a}}, \label{eq:traj0raised}
\end{equation}
which is equivalent to $ (1+ \frac{2R\delta_0}{\alpha})=2^{\frac{-2R C(R_5+2R)}{\alpha}} e^{\frac{2R\delta_0}{\alpha}}$. By raising  both sides to the power of $\alpha/(2R)$, we get $1 < (1+ \frac{2R\delta_0}{\alpha})^{\frac{\alpha}{2R}}= 2^{-C(R_5+2R)}e^{\delta_0}$. After raising it to the power $2R$ again, this can be rewritten as $(1+ \frac{\omega}{y})^y=\beta$ with $\omega = 2R\delta_0 > 0$, $y = \alpha > 0$, and $\beta = e^{2R(\delta_0-C(R_5+2R)\log(2))} > 1$. Substituting $z = 1+ \frac{\omega}{y}> 1$, we get $e^{\frac{\omega}{z-1}\log(z)}= \beta$, and taking the natural logarithm on both sides establishes
\begin{align}
\label{eq:Thm4term1}
&\log(z)=(z-1) \gamma,
\end{align}
for $\gamma = \frac{\log(\beta)}{\omega}> 0$. We need to solve \cref{eq:Thm4term1} for $z>1$ so as to obtain $\alpha=y=\frac{\omega}{z}$. 
Exponentiation of \cref{eq:Thm4term1} yields $z e^{-z \gamma}= e^{-\gamma}$, 
and multiplication by $-\gamma$ finally gives us $-z \gamma e^{- z \gamma}= - \gamma e^{-\gamma}$. We can solve this equation for $-z\gamma$ by means of the Lambert $W$
function. Since $\gamma > 0$ and we need the solution to satisfy $z > 1$, we
must take the branch $W_{-1}$ here to compute $z= - \frac{W_{-1}(-\gamma e^{-\gamma})}{\gamma}$.
Plugging in the definitions of $z$ and $\gamma$ into $y= \frac{\omega}{z}$, we obtain
\begin{align}
\label{eq:Thm4term2}
&y= - \frac{-\log(\beta)}{W_{-1}(-\frac{\log(\beta)}{\omega} \beta^{-\frac{1}{\omega}})+ \frac{\log(\beta)}{\omega}}.
\end{align}
Finally, replacing $\omega$ resp.\ $\beta$ by their ``generic'' definitions
$\omega = 2Rt$ resp.\ $\beta=e^{2R(t-C(R_5+2R)\log(2))}$ (where
$\delta_0$ is replaced by $t$) in \cref{eq:Thm4term2} gives \cref{eq:AtRC}.

It only remains to justify \cref{eq:alpha1} and \cref{eq:alpha2}, for which the same 
procedure as for \cref{eq:deltainf_int} and \cref{eq:deltaminf_int} can be used:
The same derivations as above, except that we
start from the variant of \cref{eq:traj0raised} where $a$ is replaced by
$\frac{\alpha_1}{2R}$ resp.\ $\frac{\alpha_2}{2R}$ for 
\cref{eq:alpha1} resp.\ \cref{eq:alpha2}. This finally also explains
why we can determine $R$ by (numerically) solving \cref{eq:forR}.
\end{proof}

\subsection{Experimental accuracy evaluation}
\label{Sec:Result}

In this section, we evaluate the accuracy of our interconnect-augmented model by comparing its predictions to
the actual delays of an interconnected \NOR\ gate obtained via analog simulations. 
More specifically, as illustrated in \cref{fig:experiment_setup}, we instantiated 
a \NOR\ gate connected to an inverter, acting as its load, via a controlled wire. 
The inputs of the \NOR\ gate are driven by a chain of 4 inverters acting as 
signal shaping gates. The chain input is stimulated by a saturated ramp 
with a rise/fall time of 0.1\,fs, which leads to ``natural'' signal waveforms
at the chain output. Different slew rates at the inputs of the \NOR\ gate
are generated by varying the driving strength of the last shaping inverter.

\begin{figure}[t!]
  \centering
  \includegraphics[width=0.7\linewidth]{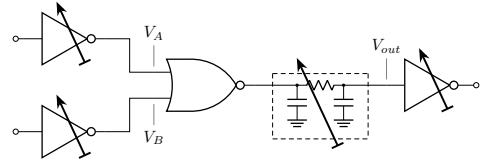}
  \caption{\small\em Experimental setup.}
  \label{fig:experiment_setup}
\end{figure}

For every setting, the following steps were performed:
\begin{enumerate}
\item[(1)] Utilizing a Verilog description of our CMOS \NOR\ gate implementation, we
employed the Cadence tools Genus and Innovus (version 19.11) for placing and routing
our design. 
\item[(2)] Using the extracted parasitic networks from the final layout, 
we performed SPICE simulations to determine
$\delta^{\uparrow / \downarrow}$ for different values of $\Delta$. 
\item[(3)] Using the measured MIS delay values $\ddoD_S(\infty)$,  $\ddoD_S(0)$, $\ddoD_S(\infty)$ 
and $\dupD_S(\infty)$, $\dupD_S(0)$, $\dupD_S(-\infty)$, as well as the measured minimal
pure delay $\dmin$ determined according to the procedure described in \cite{maier2021composable}
(with inputs $A$ and $B$ tied together) and some rough estimate\footnote{Our parametrization procedure
can adapt to any value for $C$, by scaling the resistors $R_{n_A}$, $R_{n_B}$ and $R$
appropriately.} of the load capacitance $C$,
we used \cref{thm:gatechar} for parametrizing our model.
\item[(4)] Using the equations given in \cref{Corol:1}, we computed the predictions 
of the parametrized model for different values $\Delta$,
and compared the outcome to the measured delays.
\end{enumerate}

The different settings used in the evaluation range from different 
implementation technologies to varying driving strengths and load capacitances 
to different wire lengths, wire resistances, and wire capacitances. 
Most of these results were obtained for a CMOS \NOR\ gate from the \SI{15}{\nm} 
Nangate Open Cell Library featuring FreePDK15$^\text{TM}$ FinFET models~\cite{Nangate15}
with a supply voltage of $\vdd=\SI{0.8}{\V}$. Qualitatively similar results have been obtained for
the UMC \SI{65}{\nm} technology with $\vdd=\SI{1.2}{V}$.

Overall, the accuracy of our interconnect-augmented model turned out to be surprisingly 
good, in any setting, despite our model's simplicity. Indeed, for none of the many
choices for wirelenghts etc.\ that have been explored in our experiments, i.e., not
just for the ones given below, we observed a worst-case inaccuracy above the 10\% range
in its delay predictions. Albeit this is not competitive compared to the modeling approaches 
used in STA, which achieve accuracies in the \%-range, see~\cref{sec:MIS}, it is a remarkable
improvement over the state-of-the-art in digital dynamic timing analysis: As explained in
\cref{sec:DDTA}, existing tools rely on pure or inertial delay models here, which 
do cover MIS effects at all. Note that this also explains why an explicit comparison 
to these approaches would be void.

On the other hand, we need to stress that these reassuring results are in stark contrast 
to the ones obtained for the original model \cite{ferdowsi2023accurate} (for ``naked'') 
\NOR\ gates, where even the parametrization procedure in step (3) already failed 
in most scenarios considered in this paper! This also confirms that adding $R_5$ 
is really instrumental for modeling interconnected gates.

A representative sample of our results will be presented in the following subsections.
In all our figures, the SPICE-generated delays are depiced by the dashed red curve,
whereas the delays predicted by our model are represented by the blue curve.

\begin{figure}[t]
  \centerline{
    \includegraphics[width=0.7\linewidth]{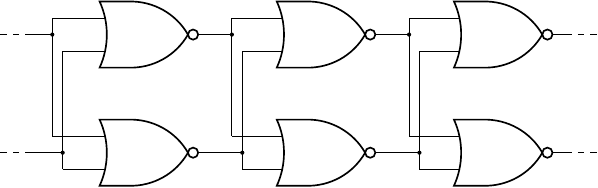}
  }
  \caption{Illustration of three stages of the cross-coupled \NOR\ gates chain used for
    simulation time evaluation.}
  \label{fig:RunTime}
\end{figure}

\medskip

In order to also experimentally confirm our (intuitively obvious) claim that 
the running time of an implementation of our model in the context of the discrete 
event simulation-based Involution Tool \invt\ \cite{OMFS20:INTEGRATION} for dynamic
digital timing analysis would outperform the numerical integration-based SPICE by 
orders of magnitude, we used the Python interface added to the \invt\ in 
\cite{OS23:DDECS} to add support for \NOR\ gates using the delay formulas 
of \cref{Corol:1}. As our target circuit, we use two parallel chains of 
$n$ identical cross-coupled
\NOR\ gates as shown in \cref{fig:RunTime}, which are
stimulated by two input signals $I_1$ and $I_2$ that are randomly
generated pulse trains according to a normal distribution with mean $\mu$ and
standard deviation $\sigma$. 

For $n=50$, we determined the simulation times (averaged
over 20 runs each) for two different settings: (a) for $N=1000$ transitions with increasing
average $\mu \in \{50, 100, 200\}$ps
and $\sigma \in \{30, 60, 120\}$ps, which amounts to two times
doubling the average length of the simulated signal traces and
(b) for $N \in \{1000, 2000, 4000\}$ transitions with decreasing
$\mu \in \{200, 100, 50\}$ps and $\sigma \in \{120, 60, 30\}$ps,
which amounts to two times doubling the number of transitions within
a constant average trace length.
\cref{table:runtime} shows the running times for both cases: The
top entry in every element of the matrix is the running time of
our Involution Tool, the bottom entry is the running time of SPICE.

\begin{table}[h]
\centering
\caption{\small\em Running time comparison (in seconds) of our model implementation 
in the \invt \ (top) vs. SPICE (bottom), for different
numbers of transitions $N$ and increasing average trace lengths (encoded via $\mu$).}
\label{table:runtime}
\scalebox{0.8}{
\begin{tabular}{|c|c|c|c|}
\hline
    & $N=1000$ & $N=2000$ & $N=4000$\\\hline
$\mu=50,\sigma=30$ & \begin{tabular}[c]{@{}c@{}}$18.9$~s\\ $941.1$~s \end{tabular} & \begin{tabular}[c]{@{}c@{}}$$\\ $$\end{tabular} & \begin{tabular}[c]{@{}c@{}}$61.5$~s\\ $3741.6$~s\end{tabular} \\ \hline
$\mu=100,\sigma=60$ & \begin{tabular}[c]{@{}c@{}}$20.3$~s\\ $1428.7$~s \end{tabular} & \begin{tabular}[c]{@{}c@{}}$35.5$~s\\ $2803.8$~s\end{tabular} & \begin{tabular}[c]{@{}c@{}}$$\\ $$\end{tabular} \\ \hline
$\mu=200,\sigma=120$ & \begin{tabular}[c]{@{}c@{}}$20.5$~s\\ $1836.4$~s \end{tabular} & \begin{tabular}[c]{@{}c@{}}$$\\ $$\end{tabular} & \begin{tabular}[c]{@{}c@{}}$$\\ $$\end{tabular} \\ \hline
\end{tabular}}
\end{table}

As expected, nonwithstanding the fact that the Involution Tool is just a
research prototype and has hence never been optimized for performance at all
(but rather slowed down substantially by incorporating Python code), it
outperforms SPICE by almost two orders of magnitude. Regarding scalability,
it is apparent from the second row (for $\mu=100$ps) that doubling the number
of transitions $N$ causes both running times to almost double (75\% resp.\ 
96\% increase for \invt\ resp.\ SPICE). The same behavior is exhibited by
\invt\ if the average trace length is kept constant when $N$ is doubled 
(the secondary diagonal in \cref{table:runtime}), albeit SPICE shows an
increase of only about 50\% resp.\ 33\% in the first and second doubling.
If only the trace length is doubled but the number of transitions $N$ 
is fixed (the column for $N=1000$), the running time of the \invt\ does 
not go up. The running time of SPICE increases, though, albeit only 
by around 50\% resp.\ 30\% for the first resp.\ the second doubling. 
We conjecture that the observed running time improvement for SPICE is caused by
its time-step adaptive numeric integration method, which speeds up in
the case of less-varying signals. 

To also demonstrate what happens when doubling the size of the circuit,
we also ran our comparison for $N=1000$, $\mu=50$ps and $\sigma=30$ps 
for a chain consisting of $n=100$ stages.
We observed a simulation time of $34.6$~s for the Involution Tool
and $2246.5$~s for SPICE. The about 83\% increase for the \invt\ compared
to the top-left entry in \cref{table:runtime} can
be traced back to the doubling of the number of transitions
occurring in the circuit, the about 183\% increase for SPICE
is primarily caused by the doubling of the number of transistors.

\medskip

\subsubsection{Wire length}
Utilizing the 15nm technology, we varied the length of the wire driven by the \NOR\ gate across the range of $l=3$ to $l=15$ micrometers\footnote{The lenght parameter $l$ actually corresponds to the parameter \$LENGTH in the command \emph{relativePlace inv1 nor1 -relation R -xOffset \$LENGTH -yOffset 0}, and thus approximately represents the length in $\mu$m, disregarding vias.}. Note that our choice of $l$ for the data shown is not really important, as we observed
very similar accuracies also for longer and shorter wires. The model parameters are given in \cref{table:Param2}, and the results are shown in \cref{Fig:Wl}. The modeling accuracy is indeed remarkable.

\begin{table}[h]
\centering
\caption{\small\em Model parameter values for two wire lengths $l=3 \  \mu m$ and $l=15 \  \mu m$, for which $\delta_{min}=4.32\ ps$ and $\delta_{min}=5.08\ ps$, respectively. The chosen load capacitance is $C=1.2831 fF$.}
\label{table:Param2}
\scalebox{0.75}{
\begin{tabular}{|clcc|}
\hline
\multicolumn{4}{|c|}{Parameters for $l=3 \  \mu m$}                                           \\ \hline
\multicolumn{2}{|c|}{$R_{n_{A}}=2.1936 \ k \si{\ohm}$} & \multicolumn{1}{c|}{$R_{n_{B}}=2.011 \ k \si{\ohm}$} & $R_5=399.41 \ \si{\ohm}$      \\ \hline
\multicolumn{2}{|c|}{$R=1.2771 \ k \si{\ohm}$}         & \multicolumn{1}{c|}{$\alpha_1=1.078 e-9 \ \si{\ohm} s$}  & $\alpha_2=0.5102 e-9 \ \si{\ohm} s$ \\ \hline
\multicolumn{4}{|c|}{Parameters for $l=15 \  \mu m$}                                          \\ \hline
\multicolumn{2}{|c|}{$R_{n_{A}}=2.9 \ k \si{\ohm}$} & \multicolumn{1}{c|}{$R_{n_{B}}=2.7493 \ k \si{\ohm}$} & $R_5=360.49 \ \si{\ohm}$      \\ \hline
\multicolumn{2}{|c|}{$R=2.0545 \ k \si{\ohm}$}         & \multicolumn{1}{c|}{$\alpha_1=1.479e-9 \ \si{\ohm} s$}  & $\alpha_2=0.8441e-9 \ \si{\ohm} s$ \\ \hline
\end{tabular}}
\end{table}

\begin{figure*}[t!]
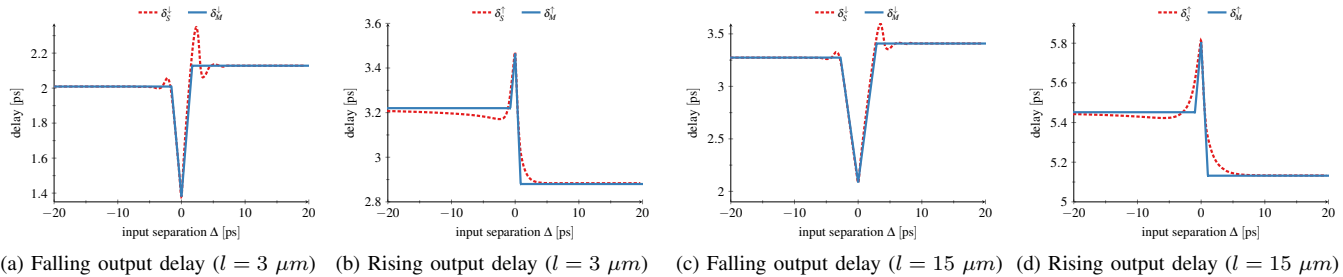

  \centering
  \subfloat[Falling output delay ($l=3 \  \mu m$)]{
    \includegraphics[width=0.23\linewidth]{\figPath{INT_RITL3.pdf}}%
    \label{Fig:WlR3}}
  \hfil
  \subfloat[Rising output delay ($l=3 \  \mu m$)]{
    \includegraphics[width=0.23\linewidth]{\figPath{INT_FITL3.pdf}}%
    \label{Fig:WlF3}}
\ifthenelse{\boolean{conference}}
{}
{ \hfil
  \subfloat[Falling output delay ($l=5 \  \mu m$)]{
 \includegraphics[width=0.23\linewidth]{\figPath{INT_RIT5.pdf}}
    \label{Fig:WlR5}}
  \hfil
  \subfloat[Rising output delay ($l=5 \  \mu m$)]{
\includegraphics[width=0.23\linewidth]{\figPath{INT_FIT5.pdf}}%
    \label{Fig:WlF5}}
  \hfil
  \subfloat[Falling output delay ($l=10 \  \mu m$)]{
    \includegraphics[width=0.23\linewidth]{\figPath{INT_RIT10.pdf}}%
    \label{Fig:WlR10}}
  \hfil
  \subfloat[Rising output delay ($l=10 \  \mu m$)]{
    \includegraphics[width=0.23\linewidth]{\figPath{INT_FIT10.pdf}}%
    \label{Fig:WlF10}}
}
 \hfil
  \subfloat[Falling output delay ($l=15 \  \mu m$)]{
 \includegraphics[width=0.23\linewidth]{\figPath{INT_RITL15.pdf}}
    \label{Fig:WlR15}}
  \hfil
  \subfloat[Rising output delay ($l=15 \  \mu m$)]{
\includegraphics[width=0.23\linewidth]{\figPath{INT_FITL15.pdf}}%
    \label{Fig:WlF10}}
  \caption{\small\em SPICE-generated ($\delta_S^{\uparrow / \downarrow}(\Delta)$) and predicted ($\delta_M^{\uparrow / \downarrow}(\Delta)$) MIS delays for a 15nm technology \NOR\ gate for different wire lengths $l$.}\label{Fig:Wl}  
\end{figure*}

\subsubsection{Wire resistance and capacitance}
In order to verify the ability of our model to adapt to varying parasitic networks, we artificially changed the resistances and capacitances of the extracted network for wire length $l=15 \  \mu m$: in one setting, we halved all the resistor values, and in another setting, we doubled the values of all capacitors.
\cref{table:Param3} gives the model parameters and \cref{result_CAP_RES_15} shows the results, again revealing
a very good match.

\begin{table}[h]
\centering
\caption{\small\em Model parameter values for different wire resistances and capacitances, with $\delta_{min}=0.51\ ps$ and $\delta_{min}=0.46\ ps$ for double capacitance and half resistance. The chosen load capacitance value is $C=1.2831 fF$.}
\label{table:Param3}
\scalebox{0.75}{
\begin{tabular}{|clcc|}
\hline
\multicolumn{4}{|c|}{Parameters for doubling the capacitance}                                           \\ \hline
\multicolumn{2}{|c|}{$R_{n_{A}}=4.5510 \ k \si{\ohm}$} & \multicolumn{1}{c|}{$R_{n_{B}}=4.3896 \ k \si{\ohm}$} & $R_5=447.11 \ \si{\ohm}$      \\ \hline
\multicolumn{2}{|c|}{$R=3.5436 \ k \si{\ohm}$}         & \multicolumn{1}{c|}{$\alpha_1=2.215 e-9 \ \si{\ohm} s$}  & $\alpha_2=1.393 e-9 \ \si{\ohm} s$ \\ \hline
\multicolumn{4}{|c|}{Parameters for half the resistor values}                                          \\ \hline
\multicolumn{2}{|c|}{$R_{n_{A}}=2.9037 \ k \si{\ohm}$} & \multicolumn{1}{c|}{$R_{n_{B}}=2.7578 \ k \si{\ohm}$} & $R_5=366.26 \ \si{\ohm}$      \\ \hline
\multicolumn{2}{|c|}{$R=2.0503 \ k \si{\ohm}$}         & \multicolumn{1}{c|}{$\alpha_1=1.486e-9 \ \si{\ohm} s$}  & $\alpha_2=0.88e-9 \ \si{\ohm} s$ \\ \hline
\end{tabular}}
\end{table}

\begin{figure*}[t!]
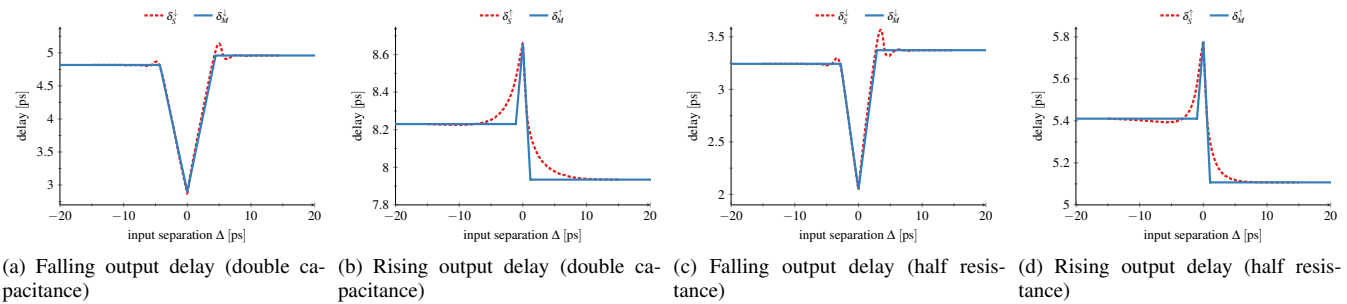

  \centering
  \subfloat[Falling output delay (double capacitance)]{
    \includegraphics[width=0.23\linewidth]{\figPath{INT_RITL15DC.pdf}}%
    \label{fig:resultDoublecapa}}
  \hfil
  \subfloat[Rising output delay (double capacitance)]{
    \includegraphics[width=0.23\linewidth]{\figPath{INT_FITL15DC.pdf}}%
    \label{fig:resultDoublecap2}}
 \hfil
  \subfloat[Falling output delay (half resistance)]{
 \includegraphics[width=0.23\linewidth]{\figPath{INT_RITL15HR.pdf}}
    \label{fig:resultHalfRa}}
  \hfil
  \subfloat[Rising output delay (half resistance)]{
\includegraphics[width=0.23\linewidth]{\figPath{INT_FITL15HR.pdf}}%
    \label{fig:resultHalfRb}}
  \caption{\small\em SPICE-generated ($\delta_S^{\uparrow / \downarrow}(\Delta)$) and predicted ($\delta_M^{\uparrow / \downarrow}(\Delta)$) MIS delays for a 15nm technology \NOR\ gate for wire length $l=15 \  \mu m$ when the wire capacitances are doubled (two left figures) resp.\ the wire resistors are halved (two right figures).}\label{result_CAP_RES_15}  
\end{figure*}

\subsubsection{Load capacitance}
Additionally, we explored varying the load capacitance of the 15nm \NOR\ gate with wire lengths of $l=3 \  \mu m$ and $l=15 \  \mu m$, achieved by increasing the fan-out of the \NOR\ gate acting as a load. To accomplish this, we used inverters comprising $2$, $4$, and $8$ parallel \pmos\ and \nmos\ transistors. The outcomes of these experiments, which used the parametrization given in \cref{table:Param4}, are depicted in \cref{result_LCAP_L3_15}. Note that we had to incorporate different (increasing) values for $C$ in each case, to match the quite different (increasing) measured delays. 

\begin{table}[h]
\centering
\caption{\small\em Model parameter values for different load capacitances.}
\label{table:Param4}
\scalebox{0.75}{
\begin{tabular}{|clccc|}
\hline
\multicolumn{5}{|c|}{Parameters for \cref{fig:Cap_resulta} and \cref{fig:Cap_resultb}.}            \\ \hline
\multicolumn{2}{|c|}{$C= 1.2831 \ fF$}   & \multicolumn{1}{c|}{$\delta_{min}= 0.41 \ ps$}       & \multicolumn{1}{c|}{$R_{n_{A}}= 2.1496  \ k \si{\ohm}$}       & $R_{n_{B}}= 2.0068  \ k \si{\ohm}$       \\ \hline
\multicolumn{2}{|c|}{$R_5= 355.82   \si{\ohm}$} & \multicolumn{1}{c|}{$R= 1.3993  \ k \si{\ohm}$}                  & \multicolumn{1}{c|}{$\alpha_1=1.075 e-9  \ \si{\ohm} s$}        & $\alpha_2= 0.564 e-9 \ \si{\ohm} s$        \\ \hline
\multicolumn{5}{|c|}{Parameters for \cref{fig:Cap_resultc} and \cref{fig:Cap_resultd}.}            \\ \hline
\multicolumn{2}{|c|}{$C= 2.9775 \ fF$}   & \multicolumn{1}{c|}{$\delta_{min}= 0.29\ ps$}       & \multicolumn{1}{c|}{$R_{n_{A}}= 2.6760  \ k \si{\ohm}$}       & $R_{n_{B}}= 2.5921 \ k \si{\ohm}$       \\ \hline
\multicolumn{2}{|c|}{$R_5= 232.77  \ \si{\ohm}$} & \multicolumn{1}{c|}{$R= 2.1640  \ k \si{\ohm}$}                  & \multicolumn{1}{c|}{$\alpha_1= 1.273 e-9 \ \si{\ohm} s$}        & $\alpha_2= 0.785 e-9 \ \si{\ohm} s$        \\ \hline
\multicolumn{5}{|c|}{Parameters for \cref{fig:Cap_resulte} and \cref{fig:Cap_resultf}.}            \\ \hline
\multicolumn{2}{|c|}{$C= 1.2831 \ fF$}   & \multicolumn{1}{c|}{$\delta_{min}= 0.46  \ ps$}       & \multicolumn{1}{c|}{$R_{n_{A}}= 3.4405  \ k \si{\ohm}$}       & $R_{n_{B}}= 3.2801 \ k \si{\ohm}$       \\ \hline
\multicolumn{2}{|c|}{$R_5= 434.88  \ \si{\ohm}$} & \multicolumn{1}{c|}{$R= 2.5447 \ k \si{\ohm}$}                  & \multicolumn{1}{c|}{$\alpha_1= 1.697 e-9 \ \si{\ohm} s$}        & $\alpha_2= 0.984 e-9 \ \si{\ohm} s$        \\ \hline
\multicolumn{5}{|c|}{Parameters for \cref{fig:Cap_resultg} and \cref{fig:Cap_resulth}.}            \\ \hline
\multicolumn{2}{|c|}{$C=3.2831 \ fF$}   & \multicolumn{1}{c|}{$\delta_{min}= 0.41 \ ps$}       & \multicolumn{1}{c|}{$R_{n_{A}}= 2.6738  \ k \si{\ohm}$}       & $R_{n_{B}}=  2.5997 \ k \si{\ohm}$       \\ \hline
\multicolumn{2}{|c|}{$R_5= 197.55  \ \si{\ohm}$} & \multicolumn{1}{c|}{$R= 2.2261 \ k \si{\ohm}$}                  & \multicolumn{1}{c|}{$\alpha_1= 1.138  e-9 \ \si{\ohm} s$}        & $\alpha_2= 0.693 e-9 \ \si{\ohm} s$        \\ \hline
\end{tabular}}
\end{table}

\begin{figure*}[t!]
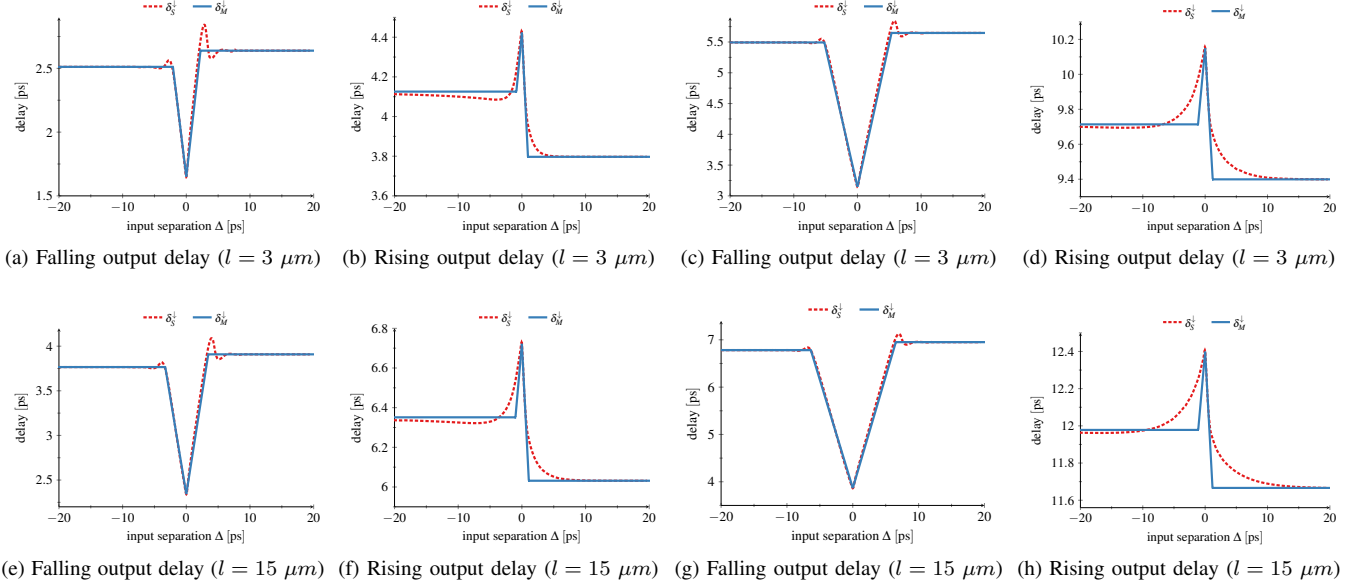

  \centering
  \subfloat[Falling output delay ($l=3 \  \mu m$)]{
    \includegraphics[width=0.23\linewidth]{\figPath{INT_RITL3_LC2.pdf}}%
    \label{fig:Cap_resulta}}
  \hfil
  \subfloat[Rising output delay ($l=3 \  \mu m$)]{
    \includegraphics[width=0.23\linewidth]{\figPath{INT_FITL3_LC2.pdf}}%
    \label{fig:Cap_resultb}}
 \hfil
  \subfloat[Falling output delay ($l=3 \  \mu m$)]{
 \includegraphics[width=0.23\linewidth]{\figPath{INT_RITL3_LC8.pdf}}
    \label{fig:Cap_resultc}}
  \hfil
  \subfloat[Rising output delay ($l=3 \  \mu m$)]{
\includegraphics[width=0.23\linewidth]{\figPath{INT_FITL3_LC8.pdf}}%
    \label{fig:Cap_resultd}}
    \hfil
     \subfloat[Falling output delay ($l=15 \  \mu m$)]{
    \includegraphics[width=0.23\linewidth]{\figPath{INT_RITL15_LC2.pdf}}%
    \label{fig:Cap_resulte}}
  \hfil
  \subfloat[Rising output delay ($l=15 \  \mu m$)]{
    \includegraphics[width=0.23\linewidth]{\figPath{INT_FITL15_LC2.pdf}}%
    \label{fig:Cap_resultf}}
 \hfil
  \subfloat[Falling output delay ($l=15 \  \mu m$)]{
 \includegraphics[width=0.23\linewidth]{\figPath{INT_RITL15_LC8.pdf}}
    \label{fig:Cap_resultg}}
  \hfil
  \subfloat[Rising output delay ($l=15 \  \mu m$)]{
\includegraphics[width=0.23\linewidth]{\figPath{INT_FITL15_LC8.pdf}}%
    \label{fig:Cap_resulth}}
  \caption{\small\em SPICE-generated ($\delta_S^{\uparrow / \downarrow}(\Delta)$) and predicted ($\delta_M^{\uparrow / \downarrow}(\Delta)$) MIS delays for a \SI{15}{\nm} technology \NOR\ gate driving two parallel pMOS and nMOS transistors (two left figures) and eight parallel pMOS and nMOS transistors (two right figures) for different wire lengths $l$.}\label{result_LCAP_L3_15}  
\end{figure*}

\subsubsection{Input gate driving strength}
In two other settings, we varied the driving strength of the two input inverters in \cref{fig:experiment_setup} that drive $V_A$ and $V_B$ of the \SI{15}{\nm} \NOR\ gate connected to the load inverter via a wire of length $l=15 \  \mu m$. More specifically, we used both a strong input inverter (comprising four parallel pMOS and nMOS transistors) and a weak input inverter (which was simulated by letting the input inverters drive three additional \NOR\ gates, resulting in a fan-out of four). \cref{result_SWI_L15} shows the outcomes for both scenarios, which have been obtained using the model parameters listed in \cref{table:Param5}.

\begin{table}[h]
\centering
\caption{\small\em Model parameter values associated with different input gate driving strength.}
\label{table:Param5}
\scalebox{0.75}{
\begin{tabular}{|clccc|}
\hline
\multicolumn{5}{|c|}{Parameters for \cref{fig:WI_resulta} and \cref{fig:WI_resultb}.}            \\ \hline
\multicolumn{2}{|c|}{$C= 1.2831 \ fF$}   & \multicolumn{1}{c|}{$\delta_{min}= 0.78 \ ps$}       & \multicolumn{1}{c|}{$R_{n_{A}}= 3.3696  \ k \si{\ohm}$}       & $R_{n_{B}}= 3.3344  \ k \si{\ohm}$       \\ \hline
\multicolumn{2}{|c|}{$R_5= 895.42   \si{\ohm}$} & \multicolumn{1}{c|}{$R= 2.1595  \ k \si{\ohm}$}                  & \multicolumn{1}{c|}{$\alpha_1=1.634 e-9  \ \si{\ohm} s$}        & $\alpha_2= 0.942 e-9 \ \si{\ohm} s$        \\ \hline
\multicolumn{5}{|c|}{Parameters for \cref{fig:SI_resulta} and \cref{fig:SI_resultb}.}            \\ \hline
\multicolumn{2}{|c|}{$C= 1.2831 \ fF$}   & \multicolumn{1}{c|}{$\delta_{min}= 0.41 \ ps$}       & \multicolumn{1}{c|}{$R_{n_{A}}= 2.8899  \ k \si{\ohm}$}       & $R_{n_{B}}= 2.7053 \ k \si{\ohm}$       \\ \hline
\multicolumn{2}{|c|}{$R_5= 231.29  \ \si{\ohm}$} & \multicolumn{1}{c|}{$R= 2.0305  \ k \si{\ohm}$}                  & \multicolumn{1}{c|}{$\alpha_1= 1.526 e-9 \ \si{\ohm} s$}        & $\alpha_2=1.083 \ \si{\ohm} s$        \\ \hline
\end{tabular}}
\end{table}

\begin{figure*}[t!]
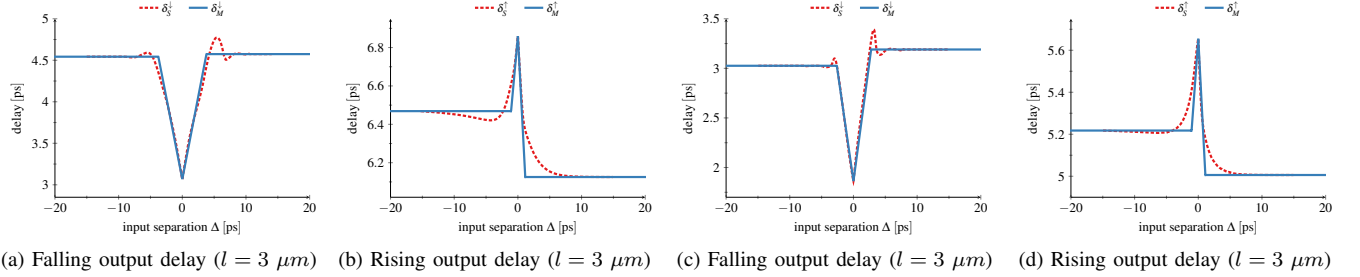

  \centering
  \subfloat[Falling output delay ($l=3 \  \mu m$)]{
    \includegraphics[width=0.23\linewidth]{\figPath{INT_RITL15WI.pdf}}%
    \label{fig:WI_resulta}}
  \hfil
  \subfloat[Rising output delay ($l=3 \  \mu m$)]{
    \includegraphics[width=0.23\linewidth]{\figPath{INT_FITL15WI.pdf}}%
    \label{fig:WI_resultb}}
 \hfil
  \subfloat[Falling output delay ($l=3 \  \mu m$)]{
 \includegraphics[width=0.23\linewidth]{\figPath{INT_RITL15SI.pdf}}
    \label{fig:SI_resulta}}
  \hfil
  \subfloat[Rising output delay ($l=3 \  \mu m$)]{
\includegraphics[width=0.23\linewidth]{\figPath{INT_FITL15SI.pdf}}%
    \label{fig:SI_resultb}}
  \caption{\small\em SPICE-generated ($\delta_S^{\uparrow / \downarrow}(\Delta)$) and predicted ($\delta_M^{\uparrow / \downarrow}(\Delta)$) MIS delays for a \SI{15}{\nm} technology \NOR\ gate for wire length $l=15 \  \mu m$ with weak input drivers (two left figures) resp. strong input drivers (two right figures).}\label{result_SWI_L15}  
\end{figure*}

\subsubsection{Other technologies}

To validate that our model achieves comparable modeling accuracies also in different technologies,
we conducted additional simulations for a \NOR\ gate in UMC \SI{65}{\nm} technology with a supply voltage of $\vdd=\SI{1.2}{V}$. Given the qualitative similarity of the results, we present only a very small subset in this paper: \Cref{result_65T} shows the results for two different wire lengths $l \in \{5,25 \}$, based on the parameters given in \cref{table:Param6}.

\begin{table}[h]
\centering
\caption{\small\em Model parameter values for different wire lengths in \SI{65}{\nm} technology.}
\label{table:Param6}
\scalebox{0.75}{
\begin{tabular}{|clccc|}
\hline
\multicolumn{5}{|c|}{Parameters for \cref{fig:WI_resulta} and \cref{fig:WI_resultb}.}            \\ \hline
\multicolumn{2}{|c|}{$C= 6.2831 \ fF$}   & \multicolumn{1}{c|}{$\delta_{min}= 1.76 \ ps$}       & \multicolumn{1}{c|}{$R_{n_{A}}= 6.2629  \ k \si{\ohm}$}       & $R_{n_{B}}= 5.8159  \ k \si{\ohm}$       \\ \hline
\multicolumn{2}{|c|}{$R_5= 4.089  \ k \si{\ohm}$} & \multicolumn{1}{c|}{$R= 600.66  \  \si{\ohm}$}                  & \multicolumn{1}{c|}{$\alpha_1=3.483 e-9  \ \si{\ohm} s$}        & $\alpha_2= 0.908 e-9 \ \si{\ohm} s$        \\ \hline
\multicolumn{5}{|c|}{Parameters for \cref{fig:SI_resulta} and \cref{fig:SI_resultb}.}            \\ \hline
\multicolumn{2}{|c|}{$C= 7.2831 \ fF$}   & \multicolumn{1}{c|}{$\delta_{min}= 1.109 \ ps$}       & \multicolumn{1}{c|}{$R_{n_{A}}= 8.1075  \ k \si{\ohm}$}       & $R_{n_{B}}= 7.7869 \ k \si{\ohm}$       \\ \hline
\multicolumn{2}{|c|}{$R_5= 4.3430  \ k  \si{\ohm}$} & \multicolumn{1}{c|}{$R= 2.065  \  \si{\ohm}$}                  & \multicolumn{1}{c|}{$\alpha_1= 9.687 e-9 \ \si{\ohm} s$}        & $\alpha_2=3.073 \ \si{\ohm} s$        \\ \hline
\end{tabular}}
\end{table}

\begin{figure*}[t!]
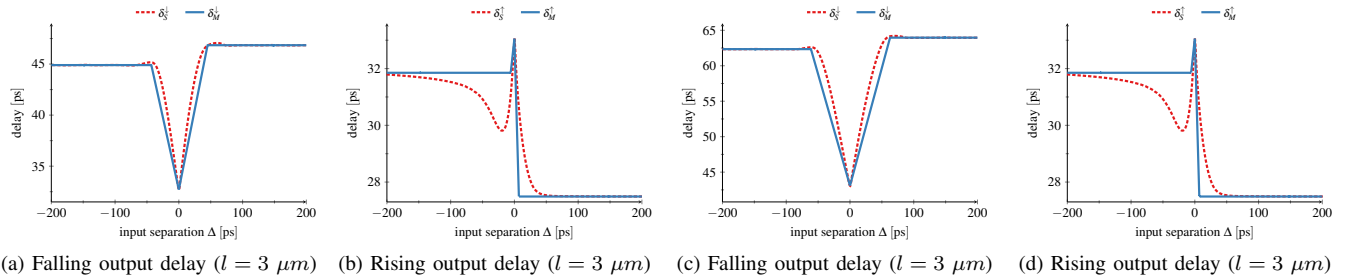

  \centering
  \subfloat[Falling output delay ($l=3 \  \mu m$)]{
    \includegraphics[width=0.23\linewidth]{\figPath{INT_RITL5_65.pdf}}%
    \label{fig:65Tl5_resulta}}
  \hfil
  \subfloat[Rising output delay ($l=3 \  \mu m$)]{
    \includegraphics[width=0.23\linewidth]{\figPath{INT_FITL5_65.pdf}}%
    \label{fig:65Tl5_resultb}}
 \hfil
  \subfloat[Falling output delay ($l=3 \  \mu m$)]{
 \includegraphics[width=0.23\linewidth]{\figPath{INT_RITL25_65.pdf}}
    \label{fig:65Tl25_resulta}}
  \hfil
  \subfloat[Rising output delay ($l=3 \  \mu m$)]{
\includegraphics[width=0.23\linewidth]{\figPath{INT_FITL5_65.pdf}}%
    \label{fig:65Tl25_resultb}}
  \caption{\small\em SPICE-generated ($\delta_S^{\uparrow / \downarrow}(\Delta)$) and predicted ($\delta_M^{\uparrow / \downarrow}(\Delta)$) MIS delays for a \SI{65}{\nm} technology \NOR\ gate with different wire length $l \in \{5 \  \mu m, 25 \  \mu m \}$.}\label{result_65T}  
\end{figure*}

\section{An Accurate Hybrid MIS Delay Model for Interconnected Muller \cg\ Gates}
\label{Sec:IntModelC}

A crucial feature of a delay modeling approach like ours is applicability to different types of
gates. As already explained in \cite{ferdowsi2023accurate}, it is easy to derive a model for 
(interconnected) \NAND\ gates from our model for \NOR\ gates: Since the former is obtained by 
simply swapping \nmos\ transistors with \pmos\ transistors and $\vdd$ with $\gnd$ and vice versa,
the appropriate formulas can be easily translated as well. More generally, our approach can principally
be applied to every gate that involves serial or parallel transistors, including \NOR, \NAND\ with more
than two inputs, as well as Muller \cg\ and \AOIgate\ (and-or-inverter). Developing the delay formulas 
for such a gate may need substantial (mathematical) effort, in particular, for more than two inputs
(this effort is inevitable for any MIS-aware delay model, cp.~e.g.\ \cite{CS96:DAC,CGB01:DAC,SRC15:TODAES,AKMK06:DAC}.

In this section, we will demonstrate this by developing a thresholded hybrid model for 
interconnected Muller \cg\ gates, which are 
core elements in asynchronous circuit designs \cite{CharlieEffect}.

\begin{figure}[th]
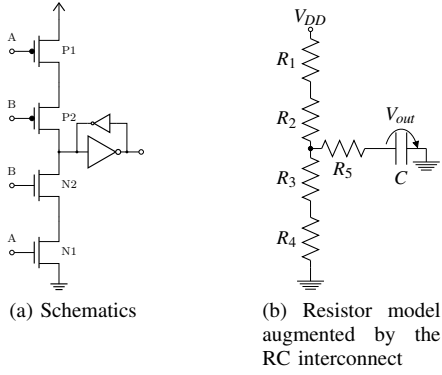

  \centering
  \subfloat[Schematics]{
\includegraphics[height=0.44\linewidth]{\figPath{Martin_C_element.pdf}}
    \label{fig:C_trans_impl}}
  \hfil
  \subfloat[Resistor model augmented by the RC interconnect]{
 \includegraphics[height=0.43\linewidth]{\figPath{C_Model_interC.pdf}}%
    \label{fig:C_Model_impl}}
  \caption{CMOS \cg\ gate implementation and the resistor model equipped with the RC interconnect component.}
\end{figure}

The thresholded hybrid model for ``naked'' Muller \cg\ gates already developed in \cite{ferdowsi2023accurate}
is based on the CMOS implementation depicted in \cref{fig:C_trans_impl}.
Despite the apparent complications introduced by the state keeper element formed by the loop of two inverters, it turned out that one can safely disregard it in the model: the load capacitance $C$ in the corresponding resistor model (\cref{fig:C_Model_impl}) effectively implements an ideal state keeper for $\vout$ when the output is in a high-impedance state, i.e., when at least one of $P1$ and $P2$ and at least one of $N1$ and $N2$ are switched off. To accommodate the negation of the output, $R_1$ and $R_2$ were designated to correspond to the \nmos\ transistors $N2$ and $N1$, and $R_3$ and $R_4$ to the \pmos\ transistors $P1$ and $P2$.
As in the case of the \NOR\ gate, continuously varying resistors according to \cref{on_mode} were assumed for switching-on, and instantaneously changing resistors for switching-off, but this time for all transistors.

Applying Kirchhoff's rules to \cref{fig:C_Model_impl} results again in the first-order non-homogeneous ODE \cref{EqIC0}, albeit with
{\small
\begin{equation}
  \frac{1}{R_g(t)}=\frac{1}{R_1(t)+R_2(t)} +\frac{1}{R_3(t)+R_4(t)}.   \label{eq:RgC}
\end{equation}
}

The approach used for solving \cref{EqIC0} and developing the appropriate delay formulas for the interconnected \NOR\ gate
can hence be adopted also here, with the 
only difference that $G(t)=\bigl( I_1(t)+I_2(t) \bigr)/C$, where $I_1(t)= \int_{0}^{t} \frac{ds}{R_1(s)+R_2(s)}$ and $I_2(t)= \int_{0}^{t} \frac{ds}{R_3(s)+R_4(s)}$. \cref{T:Cgate} summarizes all possible input state transitions, the corresponding resistor mode switch timing, the relevant integrals $I_1$ and $I_2$, and the exact or approximated value of $f(t)$.

More specifically, a detailed comparison between the transitions in \cref{T:Cgate} and those in \cref{tab:T1} and \cref{T:InerInt} shows that extracting the output voltage trajectories and deriving delay expressions for \emph{both} rising and falling input transitions for the Muller \cg\ gate is identical to the falling input transition case of the \NOR\ gate. The only differences are the substitution of $R_p$ with $R_n$ for rising input transitions, and the replacement of $\alpha_1$ resp.\ $\alpha_2$ with $\alpha_4$ resp.\ $\alpha_3$ for falling input transitions. Consequently, the delay formulas for the interconnected \cg\ gate stated in \cref{thm:MISdelayCgate} are very similar to those obtained for the falling input transition case of the interconnected \NOR\ gate.

\begin{table*}[h]
\centering
\caption{\small State transitions, integrals $I_1(t)$ and $I_2(t)$, and the function $U(t)$ for the \cg\ gate. $2R_p=R_{p_A}+R_{p_B}$ and $2R_n=R_{n_A}+R_{n_B}$.}
\scalebox{0.65}
{
\begin{tabular}{lllllllllllllllll}
\hline
Transition                &  & $t_A$       & $t_B$       &  & $R_1$                & $R_2$                & $R_3$                & $R_4$                &  & $I_1(t)= \int_{0}^{t} \frac{\dd s}{R_1(s)+R_2(s)}$                         &  & $I_2t)= \int_{0}^{t} \frac{\dd s}{R_3(s)+R_4(s)}$                          &  & $U(t)= \frac{\vdd}{C(R_1(t)+R_2(t))}$                                      &  & $f(t)=\frac{1}{\frac{R_5}{R_g(t)}+1}$ \\ \cline{1-1} \cline{3-4} \cline{6-9} \cline{11-11} \cline{13-13} \cline{15-15} \cline{17-17} 
$(0,0) \rightarrow (1,0)$ &  & $0$         & $-\infty$   &  & $on \rightarrow off$ & $on$                 & $off \rightarrow on$ & $off$                &  & $0$                                                                        &  & $0$                                                                        &  & $0$                                                                        &  & $=1$                                \\
$(1,0) \rightarrow (1,1)$ &  & $-|\Delta|$ & $0$         &  & $off$                & $on \rightarrow off$ & $on$                 & $off \rightarrow on$ &  & $\int_{0}^{t}(1/(\frac{\alpha_1}{s+\Delta}+\frac{\alpha_2}{s}+2R_n))\dd s$ &  & $0$                                                                        &  & $V_{DD}/\bigl(C( \frac{\alpha_1}{t+\Delta}+\frac{\alpha_2}{t}+2R_n)\bigr)$ &  & $\approx \frac{2R_n}{R_5+2R_n}$     \\
$(0,0) \rightarrow (0,1)$ &  & $-\infty$   & $0$         &  & $on$                 & $on \rightarrow off$ & $off$                & $off \rightarrow on$ &  & $0$                                                                        &  & $0$                                                                        &  & $0$                                                                        &  & $=1$                                \\
$(0,1) \rightarrow (1,1)$ &  & $0$         & $-|\Delta|$ &  & $on \rightarrow off$ & $off$                & $off \rightarrow on$ & $on$                 &  & $\int_{0}^{t}(1/(\frac{\alpha_1}{s}+\frac{\alpha_2}{s+\Delta}+2R_n))\dd s$ &  & $0$                                                                        &  & $V_{DD}/\bigl(C( \frac{\alpha_1}{t}+\frac{\alpha_2}{t+\Delta}+2R_n)\bigr)$ &  & $\approx \frac{2Rn}{R_5+2Rn}$        \\
$(1,1) \rightarrow (0,1)$ &  & $0$         & $-\infty$   &  & $off \rightarrow on$ & $off$                & $on \rightarrow off$ & $on$                 &  & $0$                                                                        &  & $0$                                                                        &  & $0$                                                                        &  & $=1$                                \\
$(0,1) \rightarrow (0,0)$ &  & $-|\Delta|$ & $0$         &  & $on$                 & $off \rightarrow on$ & $off$                & $on \rightarrow off$ &  & $0$                                                                        &  & $\int_{0}^{t}(1/(\frac{\alpha_3}{s}+\frac{\alpha_4}{s+\Delta}+2R_p))\dd s$ &  & $0$                                                                        &  & $\approx \frac{2R_p}{R_5+2R_p}$      \\
$(1,1) \rightarrow (1,0)$ &  & $-\infty$   & $0$         &  & $off$                & $off \rightarrow on$ & $on$                 & $on \rightarrow off$ &  & $0$                                                                        &  & $0$                                                             &  & $0$                                                                        &  & $=1$                                \\
$(1,0) \rightarrow (0,0)$ &  & $0$         & $-|\Delta|$ &  & $off \rightarrow on$ & $on$                 & $on \rightarrow off$ & $off$                &  & $0$                                                                        &  & $\int_{0}^{t}(1/(\frac{\alpha_3}{s+\Delta}+\frac{\alpha_4}{s}+2R_p))\dd s$ &  & $0$                                                                        &  & $\approx \frac{2R_p}{R_5+2R_p}$     \\ \hline
\end{tabular}}
\label{T:Cgate}
\end{table*}

\begin{thm}[MIS Delay functions for interconnected \cg\ gates] \label{thm:MISdelayCgate}
For any $0 \leq |\Delta| \leq \infty$, the MIS delay functions of our model for rising and falling input transitions are respectively given by
{\footnotesize
\begin{align}
\delta_{M,+}^{\uparrow}(\Delta) &= \begin {cases}
\delta_{0}^{\uparrow} - \frac{\alpha_1}{\alpha_1+\alpha_2} \Delta  &   \ \ 0 \leq \Delta < \frac{(\alpha_1+\alpha_2)(\delta_{0}^{\uparrow} - \delta_{\infty}^{\uparrow})}{\alpha_1}   \\ 
\delta_{\infty}^{\uparrow} &   \ \ \Delta \geq \frac{(\alpha_1+\alpha_2)(\delta_{0}^{\uparrow} - \delta_{\infty}^{\uparrow})}{\alpha_1}
\end {cases}
\\
\delta_{M,-}^{\uparrow}(\Delta) &= \begin {cases}
\delta_{0}^{\uparrow} - \frac{\alpha_2}{\alpha_1+\alpha_2} |\Delta|  &   \ \ 0 \leq |\Delta| < \frac{(\alpha_1+\alpha_2)(\delta_{0}^{\uparrow} - \delta_{-\infty}^{\uparrow})}{\alpha_2}  \\ 
\delta_{-\infty}^{\uparrow} &   \ \ |\Delta| \geq \frac{(\alpha_1+\alpha_2)(\delta_{0}^{\uparrow} - \delta_{-\infty}^{\uparrow})}{\alpha_2}
\end {cases}\label{RisingmisdelayformulaminusCG}
\end{align}
\begin{align}
\delta_{M,+}^{\downarrow}(\Delta) &= \begin {cases}
\delta_{0}^{\downarrow} - \frac{\alpha_4}{\alpha_3+\alpha_4} \Delta  &   \ \ 0 \leq \Delta < \frac{(\alpha_3+\alpha_4)(\delta_{0}^{\downarrow} - \delta_{\infty}^{\downarrow})}{\alpha_4}  \\ 
\delta_{\infty}^{\downarrow} &   \ \ \Delta \geq \frac{(\alpha_3+\alpha_4)(\delta_{0}^{\downarrow} - \delta_{\infty}^{\downarrow})}{\alpha_4}
\end {cases}\label{Fallingmisdelayformula}
\\
\delta_{M,-}^{\downarrow}(\Delta) &= \begin {cases}
\delta_{0}^{\downarrow} - \frac{\alpha_3}{\alpha_3+\alpha_4} |\Delta|  &   \ \ 0 \leq |\Delta| < \frac{(\alpha_3+\alpha_4)(\delta_{0}^{\downarrow} - \delta_{-\infty}^{\downarrow})}{\alpha_3}   \\ 
\delta_{-\infty}^{\downarrow} &   \ \ |\Delta| \geq \frac{(\alpha_3+\alpha_4)(\delta_{0}^{\downarrow} - \delta_{-\infty}^{\downarrow})}{\alpha_3}
\end {cases}\label{Fallingmisdelayformulaminus}
\end{align}
}
where
{\footnotesize
\begin{align}
\delta_{0}^{\uparrow} &= - \frac{\alpha_1 + \alpha_2}{2R_{n}} \Bigl[ 1+ W_{-1}\Bigl(\frac{-1}{e \cdot 2^{\frac{2R_{n}C(R_5+2R_{n})}{\alpha_1+ \alpha_2}}}\Bigr) \Bigr],  \label{eq:delta0Cup}
\end{align}
\begin{align}
\delta_{\infty}^{\uparrow}&= -\frac{\alpha_2}{2R_{n}} \Bigl[ 1+ W_{-1}\Bigl(\frac{-1}{e \cdot 2^{\frac{2R_{n}C(R_5+2R_{n})}{\alpha_2}}}\Bigr) \Bigr],  \label{eq:deltainfCup}
\end{align}
\begin{align}
\delta_{-\infty}^{\uparrow}&= -\frac{\alpha_1}{2R_{n}} \Bigl[ 1+ W_{-1}\Bigl(\frac{-1}{e \cdot 2^{\frac{2R_{n}C(R_5+2R_{n})}{\alpha_1}}}\Bigr) \Bigr], \label{eq:deltaminfCup}
\end{align}
\begin{align}
\delta_{0}^{\downarrow} &= - \frac{\alpha_3 + \alpha_4}{2R_{p}} \Bigl[ 1+ W_{-1}\Bigl(\frac{-1}{e \cdot 2^{\frac{2R_{p}C(R_5+2R_{p})}{\alpha_3+ \alpha_4}}}\Bigr) \Bigr],  \label{eq:delta0Cdown}
\end{align}
\begin{align}
\delta_{\infty}^{\downarrow}&= -\frac{\alpha_3}{2R_{p}} \Bigl[ 1+ W_{-1}\Bigl(\frac{-1}{e \cdot 2^{\frac{2R_{p}C(R_5+2R_{p})}{\alpha_3}}}\Bigr) \Bigr],  \label{eq:deltainfCdown}
\end{align}
\begin{align}
\delta_{-\infty}^{\downarrow}&= -\frac{\alpha_4}{2R_{p}} \Bigl[ 1+ W_{-1}\Bigl(\frac{-1}{e \cdot 2^{\frac{2R_{p}C(R_5+2R_{p})}{\alpha_4}}}\Bigr) \Bigr]. \label{eq:deltaminfCdown}
\end{align}
}
\end{thm}

\subsection{Model parametrization}

Due to the similarity of both rising and falling input transitions in our \cg\ gate modeling to the falling input transition case of the \NOR\ gate, it turns out that the parametrization for the \cg\ gate involves determining $R_n$ and $R_p$ using two functions, both of which are in the form of \cref{eq:AtRC}. The following \cref{thm:Cgatechar} provides the detailed procedure.

\begin{thm}[Model parametrization for interconnect-augmented \cg\ gates]\label{thm:Cgatechar}
Let $\dmin\geq0$ be some interconnect pure delay and $\ddoD_S(-\infty)$,  $\ddoD_S(0)$, $\ddoD_S(\infty)$ and $\dupD_S(-\infty)$,  $\dupD_S(0)$, $\dupD_S(\infty)$ be the extremal MIS delay values of a real interconnected \cg\ gate that shall be matched by our model. Given an arbitrary chosen value $C$ for the load capacitance, this is accomplished by choosing the model parameters as follows:
For

{\small
\begin{align}
&B(t,z,C)= \nonumber \\
&\qquad \frac{-\bigl(t-Cz\log(2) \bigr)}{W_{-1}\Bigl( \bigl(\frac{Cz\log(2)}{t}-1\bigr) e^{\frac{Cz\log(2)}{t}-1} \Bigl) +1 - \frac{Cz\log(2)}{t}}\label{eq:BtRC_CG},
\end{align}
}
numerically solve
\begin{align}
B\bigl(\dupD_S(0)-\dmin,x,C\bigr)& - B\bigl(\dupD_S(\infty)-\dmin,x,C\bigr)\nonumber\\
 & - B\bigl(\dupD_S(-\infty)-\dmin,x,C\bigr) = 0 \label{eq:forR_C1},\\
B\bigl(\ddoD_S(0)-\dmin,y,C\bigr)& -B\bigl(\ddoD_S(\infty)-\dmin,y,C\bigr)\nonumber\\
& - B\bigl(\ddoD_S(-\infty)-\dmin,y,C\bigr) = 0 \label{eq:forR_C2}.
\end{align}
for $x$ and $y$, respectively. For any choice of $0 \leq R_5 < \min\{x,y\}$, 
let $R_n=(x-R_5)/2>0$ and $R_p=(y-R_5)/2>0$, and
choose
\begin{flalign}
\alpha_1 &= 2R_n B\bigl(\dupD_S(-\infty)-\dmin,x,C\bigr)\label{eq:alpha1C},\\
\alpha_2 &= 2R_n B\bigl(\dupD_S(\infty)-\dmin,x,C\bigr)\label{eq:alpha2C}, \\
\alpha_3 &= 2R_p B\bigl(\ddoD_S(-\infty)-\dmin,y,C\bigr)\label{eq:alpha3C},\\
\alpha_4 &= 2R_p B\bigl(\ddoD_S(\infty)-\dmin,y,C\bigr)\label{eq:alpha4C}.
\end{flalign}
\end{thm}

\begin{proof} 
The correspondence of both the rising and falling input transition case of the \cg\ gate to the falling input transition case 
allows us to re-use \cref{eq:AtRC} for the interconnected \NOR\ gate, i.e.,

{\scriptsize
\begin{align}
&A(t,R,R_5,C)=& \nonumber \\
&\frac{-2R \bigl(t-C(R_5+2R) \cdot \log(2) \bigr)}{W_{-1}\Bigl( \bigl(\frac{C(R_5+2R) \cdot \log(2)}{t}-1\bigr) e^{\frac{C(R_5+2R) \cdot \log(2)}{t}-1} \Bigl) +1 - \frac{C(R_5+2R) \cdot \log(2)}{t}}\label{eq:AtRC_CG},&
\end{align}
}
The definition of \cref{eq:BtRC_CG} ensures
\begin{align}
A(t,R_n,R_5,C) &= \frac{B(t,R_5+2R_n,C)}{2R_n},\label{eq:AtRC_CG1}\\
A(t,R_p,R_5,C) &= \frac{B(t,R_5+2R_p,C)}{2R_p}.\label{eq:AtRC_CG2}
\end{align}

Whereas we cannot numerically solve the analog of \cref{eq:forR} for both functions 
since $R_n$, $R_p$ and $R_5$ are unknown, we can numerically solve the equivalent
equations \cref{eq:forR_C1} for $x$ and \cref{eq:forR_C2} for $y$.

Since $x=R_5+2R_n$ and $y=R_5+2R_p$, one can choose $R_5$ freely within $[0,\min\{x,y\})$. Using the resulting value
of $R_n$ resp.\ $R_p$ and recalling \cref{eq:alpha1} and \cref{eq:alpha2} yields \cref{eq:alpha1C} and \cref{eq:alpha2C}
resp.\ \cref{eq:alpha3C} and \cref{eq:alpha4C}.
\end{proof}

\subsection{Experimental accuracy evaluation}

We first consider an isolated 15nm \cg\ gate with parameters shown in \cref{table:Param3c}, with $R_5=0$. 
\cref{fig:charlieC_RIS} and \cref{fig:charlieC_FALL} show the results, which reveal
a very good match between the SPICE-generated and predicted delays.

\begin{table}[h]
\centering
\caption{\small\em Model parameter values for the 15nm \cg\ gate with chosen value of $C=2.6331 fF$, $\delta_{min}=1.77e-12$, and $R_5=0$.}
\label{table:Param3c}
\scalebox{0.8}{
\begin{tabular}{|clcc|}
\hline
\multicolumn{2}{|c|}{$R_{n}=2.1420\ k \si{\ohm}$} & \multicolumn{1}{c|}{$\alpha_1=2.1472\ \si{\ohm\s}$}  & $\alpha_2=1.1303\ \si{\ohm\s}$  \\ \hline
\multicolumn{2}{|c|}{$R_{p}=2.3215\ k \si{\ohm}$} & \multicolumn{1}{c|}{$\alpha_3= 1.5549\ \si{\ohm\s}$} & $\alpha_4= 1.8403\ \si{\ohm\s}$ \\ \hline
\end{tabular}}
\end{table}

\begin{figure}[th]
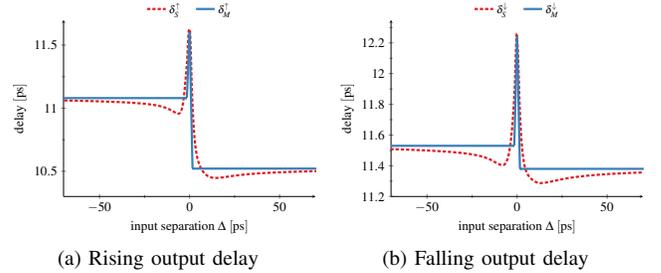

  \centering
  \subfloat[Rising output delay]{
 \includegraphics[height=0.36\linewidth]{\figPath{MIS_Cgate_RIT.pdf}}%
    \label{fig:charlieC_RIS}} 
      \hfil
  \subfloat[Falling output delay]{
 \includegraphics[height=0.36\linewidth]{\figPath{MIS_Cgate_FIT.pdf}}%
    \label{fig:charlieC_FALL}}
  \caption{Computed ($\delta_M^{\uparrow/ \downarrow}(\Delta)$) and measured ($\delta_S^{\uparrow/ \downarrow}(\Delta)$) MIS delays for a 15nm technology isolated \cg\ gate.}
\end{figure}

For interconnected \cg\ gates, we showcase some results by varying the length, resistance, and capacitance of the wire. \cref{table:ParamCgW} and \cref{table:ParamCgDCHR} present the relevant parameters for different wire lengths and different resistance/capacitance values of the wire, for essentially random choices of $R_5$. To demonstrate that the choice of $R_5$ (within the bounds stated in \cref{thm:Cgatechar}) is irrelevant for the modeling accuracy, \cref{table:ParamCgWzero} and \cref{table:ParamCgDCHRzero} provide the analogous parametrizations for the choice $R_5=0$.

\cref{Fig:WlCg} and \cref{Cresult_CAP_RES_15} show the corresponding results, for both choices of $R_5$, which again
show a very good match.

\begin{table}[h]
\centering
\caption{\small\em Model parameter values for two wire lengths $l=3 \  \mu m$ and $l=15 \  \mu m$, for \cg\ gates with $\delta_{min}=1.7\ ps$ and $\delta_{min}=1.77\ ps$, respectively. The chosen load capacitance values is $C=2.6331 fF$, the chosen value for $R_5$ given below is essentially random.}
\label{table:ParamCgW}
\scalebox{0.68}{
\begin{tabular}{|clccc|}
\hline
\multicolumn{5}{|c|}{Parameters for $l=3 \  \mu m$}                                                                            \\ \hline
\multicolumn{2}{|c|}{$R_{n}=964.76 \ \si{\ohm}$}    & \multicolumn{1}{c|}{$R_{p}=1.146 \ k \si{\ohm}$}    & \multicolumn{2}{c|}{$R_5= 545.49 \ \si{\ohm}$}                    \\ \hline
\multicolumn{2}{|c|}{$\alpha_1=645.48e-12 \ \si{\ohm} s$} & \multicolumn{1}{c|}{$\alpha_2= 264.94e-12 \ \si{\ohm} s$} & \multicolumn{1}{c|}{$\alpha_3= 255.59e-12 \ \si{\ohm} s$} & $\alpha_4=406.81e-12 \ \si{\ohm} s$ \\ \hline
\multicolumn{5}{|c|}{Parameters for $l=15 \  \mu m$}                                                                           \\ \hline
\multicolumn{2}{|c|}{$R_{n}=1.418 \ k \si{\ohm}$}    & \multicolumn{1}{c|}{$R_{p}= 1.487 \ k \si{\ohm}$}    & \multicolumn{2}{c|}{$R_5=801.28 \ \si{\ohm}$}                    \\ \hline
\multicolumn{2}{|c|}{$\alpha_1=991.49e-12 \ \si{\ohm} s$} & \multicolumn{1}{c|}{$\alpha_2=396.18e-12 \ \si{\ohm} s$} & \multicolumn{1}{c|}{$\alpha_3=942.30e-12 \ \si{\ohm} s$} & $\alpha_4=1.20e-12 \ \si{\ohm} s$ \\ \hline
\end{tabular}}
\end{table}

\begin{table}[h]
\centering
\caption{\small\em Model parameter values for two wire lengths $l=3 \  \mu m$ and $l=15 \  \mu m$, for \cg\ gates with $\delta_{min}=1.7\ ps$ and $\delta_{min}=1.77\ ps$, respectively. The chosen load capacitance values is $C=2.6331 fF$, and $R_5=0$.}
\label{table:ParamCgWzero}
\scalebox{0.66}{
\begin{tabular}{|clccc|}
\hline
\multicolumn{5}{|c|}{Parameters for $l=3 \  \mu m$}                                                                            \\ \hline
\multicolumn{2}{|c|}{$R_{n}=1.237 \ k  \si{\ohm}$}    & \multicolumn{1}{c|}{$R_{p}=1.419 \ k \si{\ohm}$}    & \multicolumn{2}{c|}{$R_5= 0 \ \si{\ohm}$}                    \\ \hline
\multicolumn{2}{|c|}{$\alpha_1=827.97e-12 \ \si{\ohm} s$} & \multicolumn{1}{c|}{$\alpha_2= 339.84e-12 \ \si{\ohm} s$} & \multicolumn{1}{c|}{$\alpha_3= 316.40e-12 \ \si{\ohm} s$} & $\alpha_4=503.61e-12 \ \si{\ohm} s$ \\ \hline
\multicolumn{5}{|c|}{Parameters for $l=15 \  \mu m$}                                                                           \\ \hline
\multicolumn{2}{|c|}{$R_{n}=1.818 \ k \si{\ohm}$}    & \multicolumn{1}{c|}{$R_{p}= 1.888 \ k \si{\ohm}$}    & \multicolumn{2}{c|}{$R_5=0 \ \si{\ohm}$}                    \\ \hline
\multicolumn{2}{|c|}{$\alpha_1=1000.27e-12 \ \si{\ohm} s$} & \multicolumn{1}{c|}{$\alpha_2=508.10e-12 \ \si{\ohm} s$} & \multicolumn{1}{c|}{$\alpha_3=1000.19e-12 \ \si{\ohm} s$} & $\alpha_4=1000.52e-12 \ \si{\ohm} s$ \\ \hline
\end{tabular}}
\end{table}

\begin{figure*}[t!]
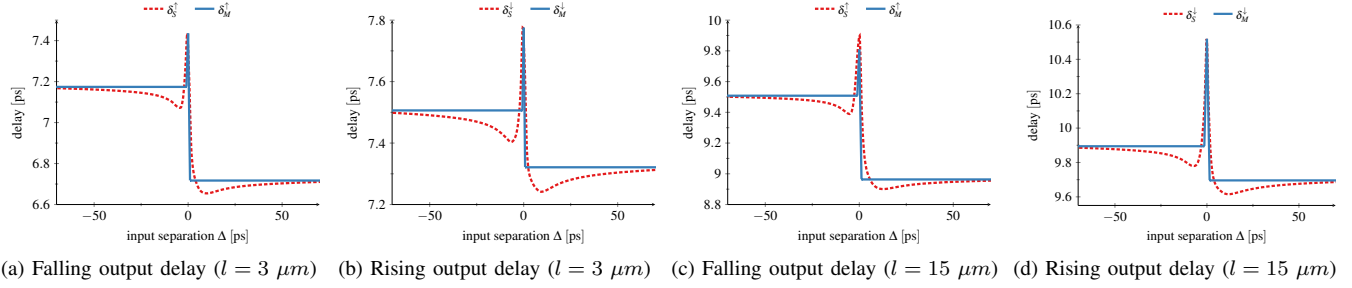

  \centering
  \subfloat[Falling output delay ($l=3 \  \mu m$)]{
    \includegraphics[width=0.23\linewidth]{\figPath{Cg_RITL3.pdf}}%
    \label{Fig:WlR3C}}
  \hfil
  \subfloat[Rising output delay ($l=3 \  \mu m$)]{
    \includegraphics[width=0.23\linewidth]{\figPath{Cg_FITL3.pdf}}%
    \label{Fig:WlF3C}}
 \hfil
  \subfloat[Falling output delay ($l=15 \  \mu m$)]{
 \includegraphics[width=0.23\linewidth]{\figPath{Cg_RITL15.pdf}}
    \label{Fig:WlR15C}}
  \hfil
  \subfloat[Rising output delay ($l=15 \  \mu m$)]{
\includegraphics[width=0.23\linewidth]{\figPath{Cg_FITL15.pdf}}%
    \label{Fig:WlF15C}}
  \caption{\small\em SPICE-generated ($\delta_S^{\uparrow / \downarrow}(\Delta)$) and predicted ($\delta_M^{\uparrow / \downarrow}(\Delta)$) MIS delays for a 15nm technology \cg\ gate for different wire lengths $l$, for any feasible choice of $R_5$.}\label{Fig:WlCg}  
\end{figure*}

\begin{table}[h]
\centering
\caption{\small\em Model parameter values for different wire resistances and capacitances of the \cg\ gate with chosen values of $\delta_{min}=1.3\ ps$ and $\delta_{min}=1.74\ ps$ for double capacitance and half resistance. The chosen load capacitance values is $C=2.6331 fF$, the chosen value for $R_5$ given below is essentially random.}
\label{table:ParamCgDCHR}
\scalebox{0.72}{
\begin{tabular}{|clccc|}
\hline
\multicolumn{5}{|c|}{Parameters for doubling the capacitance}                                                                            \\ \hline
\multicolumn{2}{|c|}{$R_{n}=1.794 \ k \si{\ohm}$}    & \multicolumn{1}{c|}{$R_{p}=2.060 \ k \si{\ohm}$}    & \multicolumn{2}{c|}{$R_5=  1.013 \ k \si{\ohm}$}                    \\ \hline
\multicolumn{2}{|c|}{$\alpha_1=2.614e-9 \ \si{\ohm} s$} & \multicolumn{1}{c|}{$\alpha_2= 1.629e-19 \ \si{\ohm} s$} & \multicolumn{1}{c|}{$\alpha_3= 1.728e-9 \ \si{\ohm} s$} & $\alpha_4=1.912e-9 \ \si{\ohm} s$ \\ \hline
\multicolumn{5}{|c|}{Parameters for half the resistor values}                                                                           \\ \hline
\multicolumn{2}{|c|}{$R_{n}=1.383\ k \si{\ohm}$}    & \multicolumn{1}{c|}{$R_{p}= 1.492 \ k \si{\ohm}$}    & \multicolumn{2}{c|}{$R_5=781.49 \ \si{\ohm}$}                    \\ \hline
\multicolumn{2}{|c|}{$\alpha_1=1.138e-9 \ \si{\ohm} s$} & \multicolumn{1}{c|}{$\alpha_2=0.516e-9 \ \si{\ohm} s$} & \multicolumn{1}{c|}{$\alpha_3=0.931e-9 \ \si{\ohm} s$} & $\alpha_4=1.184e-9 \ \si{\ohm} s$ \\ \hline
\end{tabular}}
\end{table}

\begin{table}[h]
\centering
\caption{\small\em Model parameter values for different wire resistances and capacitances of the \cg\ gate with chosen values of $\delta_{min}=1.3\ ps$ and $\delta_{min}=1.74\ ps$ for double capacitance and half resistance. The chosen load capacitance values is $C=2.6331 fF$, and $R_5=0$.}
\label{table:ParamCgDCHRzero}
\scalebox{0.72}{
\begin{tabular}{|clccc|}
\hline
\multicolumn{5}{|c|}{Parameters for doubling the capacitance}                                                                            \\ \hline
\multicolumn{2}{|c|}{$R_{n}=2.301 \ k \si{\ohm}$}    & \multicolumn{1}{c|}{$R_{p}=2.567 \ k \si{\ohm}$}    & \multicolumn{2}{c|}{$R_5=  0 \  \si{\ohm}$}                    \\ \hline
\multicolumn{2}{|c|}{$\alpha_1=3.353e-9 \ \si{\ohm} s$} & \multicolumn{1}{c|}{$\alpha_2= 2.089e-19 \ \si{\ohm} s$} & \multicolumn{1}{c|}{$\alpha_3= 2.154e-9 \ \si{\ohm} s$} & $\alpha_4=2.383e-9 \ \si{\ohm} s$ \\ \hline
\multicolumn{5}{|c|}{Parameters for half the resistor values}                                                                           \\ \hline
\multicolumn{2}{|c|}{$R_{n}=1.773\ k \si{\ohm}$}    & \multicolumn{1}{c|}{$R_{p}=1.882 \ k \si{\ohm}$}    & \multicolumn{2}{c|}{$R_5=0 \ \si{\ohm}$}                    \\ \hline
\multicolumn{2}{|c|}{$\alpha_1=1.459e-9 \ \si{\ohm} s$} & \multicolumn{1}{c|}{$\alpha_2= 662.113e-9 \ \si{\ohm} s$} & \multicolumn{1}{c|}{$\alpha_3= 1.176e-9 \ \si{\ohm} s$} & $\alpha_4= 1.495e-9 \ \si{\ohm} s$ \\ \hline
\end{tabular}}
\end{table}

\begin{figure*}[t!]
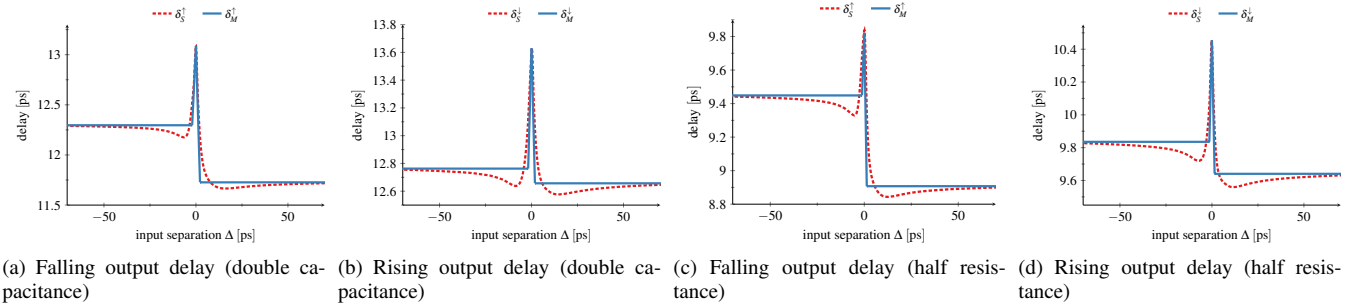

  \centering
  \subfloat[Falling output delay (double capacitance)]{
    \includegraphics[width=0.23\linewidth]{\figPath{Cg_RITL15_DC.pdf}}%
    \label{fig:CresultDoublecapa}}
  \hfil
  \subfloat[Rising output delay (double capacitance)]{
    \includegraphics[width=0.23\linewidth]{\figPath{Cg_FITL15_DC.pdf}}%
    \label{fig:CresultDoublecap2}}
 \hfil
  \subfloat[Falling output delay (half resistance)]{
 \includegraphics[width=0.23\linewidth]{\figPath{Cg_RITL15_HR.pdf}}
    \label{fig:CresultHalfRa}}
  \hfil
  \subfloat[Rising output delay (half resistance)]{
\includegraphics[width=0.23\linewidth]{\figPath{Cg_FITL15_HR.pdf}}%
    \label{fig:CresultHalfRb}}
  \caption{\small\em SPICE-generated ($\delta_S^{\uparrow / \downarrow}(\Delta)$) and predicted ($\delta_M^{\uparrow / \downarrow}(\Delta)$) MIS delays for a 15nm technology \cg\ gate for wire length $l=15 \  \mu m$ when the wire capacitances are doubled (two left figures) resp.\ the wire resistors are halved (two right figures), for any feasible choice of $R_5$.}\label{Cresult_CAP_RES_15}  
\end{figure*}

\section{Concluding Remarks and Future Work}
\label{Sec:con}

In this paper, we developed thresholded first-order hybrid delay models for interconnected multi-input gates, in particular, 2-input
\NOR\ and Muller \cg\ gates, which accurately capture multi-input switching (MIS) effects. Besides analytic formulas 
for all MIS delays in terms of the parameters, which facilitate fast digital dynamic timing analysis based on discrete event simulation, we also provided fast procedures for determining the model parameters that allow our models to match the extremal MIS delays of a given real circuit. By comparing our model predictions to SPICE simulation data, we demonstrated a surprisingly good modeling accuracy for a wide range of settings, including varying wire lengths, resistances/capacitances, input driving strengths, and output load capacitances, for two different CMOS technologies. In terms of simulation running times, our approach is orders of magnitude faster
than SPICE.

Whereas the accuracy provided by our models is an impressive improvement of the current
state-of-the-art digital dynamic timing analysis tools, where MIS effects are not considered
at all, it is by no means competitive to the delay modeling approaches used in state-of-the-art 
static timing analysis tools like PrimeTime. This also explains why we could safely ignore 
subtle effects like crosstalk and multi-input glitches in our modeling. The nevertheless
surprisingly good MIS delay prediction accuracy of our first-oder model seems to be a 
consequence of two facts, namely, the choice of a ``just right'' model (based on the 
Shichman-Hodges transistor model), and our ``end-to-end'' 
parametrization procedure, which makes sure that the model predictions match the 
``extremal'' MIS delays for any given gate.

Part of our current/future work is to fully incorporate our models into the Involution Tool, which
requires the development of extended delay formulas that also incorporate the drafting effect.
Needless to say, this extension builds on the interconnect-augmented model developed in this paper.
This will finally allow us to simulate representative benchmark circuits in the Involution Tool and fairly compare the modeling accuracy of our final model with the existing digital dynamic timing 
analysis approaches.

\bibliographystyle{IEEEtran}
\bibliography{Paper}

\end{document}